%% file: main.tex
\documentclass[letterpaper]{article}

\PassOptionsToPackage{numbers}{natbib}
\usepackage[nolinenums]{custom_arxiv}

\usepackage[utf8]{inputenc} 
\usepackage[T1]{fontenc}    
\usepackage{hyperref}       
\usepackage{url}            
\usepackage{booktabs}       
\usepackage{amsfonts}       
\usepackage{nicefrac}       
\usepackage{microtype}      

\usepackage{amssymb}
\usepackage{mathtools}
\usepackage{amsthm}
\usepackage{enumitem}
\usepackage{wrapfig}
\usepackage{float}
\usepackage{tikz}
\usepackage{thm-restate}
\usepackage{dsfont}
\usepackage{algorithm}
\usepackage[noend]{algpseudocode}
\usepackage{multirow,array}
\usepackage{wrapfig}

\newtheorem{theorem}{Theorem}

\newtheorem{lemma}[theorem]{Lemma}	
\newtheorem{definition}{Definition}

\DeclareMathOperator*{\argmin}{arg\,min}
\DeclareMathOperator*{\argmax}{arg\,max}

\newcommand{\pl}{\mathcal{P}}

\newcommand{\I}{\mathcal{I}}
\newcommand{\X}{\mathcal{X}}

\title{Learning to Correlate in Multi-Player General-Sum Sequential Games}

\author {
	{\large Andrea Celli, \, Alberto Marchesi, \, Tommaso Bianchi, \, Nicola Gatti} \\
	DEIB, Politecnico di Milano, Piazza Leonardo da Vinci 32, Milan, Italy \\
	\texttt{\{andrea.celli,alberto.marchesi,nicola.gatti\}@polimi.it,} \\
	\texttt{tommaso4.bianchi@mail.polimi.it} \\
}

\begin{document}
	
	\twocolumn[
        \maketitle
    ]
	
	\input{abstract}
	\input{content/introduction}

	\input{content/preliminaries}

	\input{content/inapx}

	\input{content/cfr_s}
	\input{content/cfr_jr}

	\input{content/experimental}
	\input{content/discussion}
	
	
	
	
	
	\bibliographystyle{custom_arxiv}
	\bibliography{main}
	
	\clearpage
	\input{content/appendix}
\end{document}

%% file: abstract.tex
\begin{abstract}
In the context of multi-player, general-sum games, there is an increasing interest in solution concepts modeling some form of communication among players, since they can lead to socially better outcomes with respect to Nash equilibria, and may be reached through learning dynamics in a decentralized fashion.
In this paper, we focus on \emph{coarse correlated equilibria} (CCEs) in sequential games.
First, we complete the picture on the complexity of finding social-welfare-maximizing CCEs by showing that the problem is not in Poly-\textsf{APX} unless \textsf{P} = \textsf{NP}.
Furthermore, simple arguments show that CFR---working with behavioral strategies---may not converge to a CCE.
However, we devise a simple variant (CFR-S) which provably converges to the set of CCEs, but may be empirically inefficient.
Thus, we design a variant of the CFR algorithm (called CFR-Jr) which approaches the set of CCEs with a regret bound sub-linear in the size of the game, and is shown to be dramatically faster than CFR-S and the state-of-the-art algorithms to compute CCEs.

\end{abstract}

%% file: content/introduction.tex
\section{Introduction}\label{sec:intro}


A number of recent studies explore relaxations of the classical notion of equilibrium (\emph{i.e.}, \emph{Nash equilibrium} (NE)~\cite{nash1951non}), allowing to model communication among the players~\cite{barman,farina2018exante,roughgarden2009intrinsic}.
%
%
Communication naturally brings about the possibility of playing correlated strategies.
%
These are customarily modeled through a trusted external mediator who privately recommends actions to the players~\cite{aumann1974subjectivity}.
%
%
%
In particular, a correlated strategy is a \emph{correlated equilibrium} (CE) if each player has no incentive to deviate from the recommendation, assuming the other players would not deviate either.
A popular variation of the CE is the \emph{coarse correlated equilibrium} (CCE), which only prevents deviations before knowing the recommendation~\cite{moulin1978}.
In sequential games, CEs and CCEs are well-suited for scenarios where the players have limited communication capabilities and can only communicate before the game starts,
%
such as, \emph{e.g.}, military settings where field units have no time or means of communicating during a battle, collusion in auctions where communication is illegal during bidding, and, in general, any setting with costly communication channels or blocking environments.

CCEs present a number of appealing properties. 
A CCE can be reached through simple (no-regret) learning dynamics in a decentralized fashion~\cite{foster1997calibrated,hart2000simple}, and, in several classes of games (such as, \emph{e.g.}, normal-form and succinct games~\cite{papadimitriou2008,jiang2015polynomial}), it can be computed exactly in time polynomial in the size of the input.
Furthermore, an optimal (\emph{i.e.}, social-welfare-maximizing) CCE may provide arbitrarily larger welfare than an optimal CE, which, in turn, may provide arbitrarily better welfare than an optimal NE~\cite{celli2018computing}.
Although the problem of finding an optimal CCE is \textsf{NP}-hard for some game classes (such as, \emph{e.g.}, graphical, polymatrix, congestion, and anonymous games~\cite{barman}), 
\citet{roughgarden2009intrinsic} shows that the CCEs reached through regret-minimizing procedures have near-optimal social welfare when the $(\lambda,\mu)$-smoothness condition holds.
This happens, \emph{e.g.}, in some specific auctions, congestion games, and even in Bayesian settings, as showed by \citet{hartline2015no}).
Thus, decentralized computation via learning dynamics, computational efficiency, and welfare optimality make the CCE one of the most interesting solution concepts for practical applications.
However, the problem of computing CCEs has been addressed only for some specific games with particular structures~\cite{barman,hartline2015no}.
In this work, we study how to compute CCEs in the general class of games which are sequential, general-sum, and multi-player.
This is a crucial advancement of CCE computation, as sequential games provide a model for strategic interactions which is richer and more adherent to real-world situations than the normal form.

In sequential games, it is known that, when there are two players without chance moves, an optimal CCE can be computed in polynomial time~\cite{celli2018computing}.
\citet{celli2018computing} also provide an algorithm (with no polynomiality guarantees) to compute solutions in multi-player games, using a column-generation procedure with a MILP pricing oracle. 
%
As for computing approximate CCEs, in the normal-form setting, any  \emph{Hannan consistent} regret-minimizing procedure for simplex decision spaces may be employed to approach the set of CCEs~\cite{blum2007learning,cesa2006prediction}---the most common of such techniques is \emph{regret matching} (RM)~\cite{blackwell1956analog,hart2000simple}.
However, approaching the set of CCEs in sequential games is more demanding. 
One could represent the sequential game with its equivalent normal form and apply RM to it. However, this would result in a guarantee on the cumulative regret which would be exponential in the size of the game tree (see Section~\ref{sec:preliminaries}).
Thus, reaching a good approximation of a CCE could require an exponential number of iterations.
The problem of designing learning algorithms avoiding the construction of the normal form has been successfully addressed in sequential games for the two-player, zero-sum setting.
This is done by decomposing the overall regret locally at the information sets of the game~\cite{farina2018online}.
The most widely adopted of such approaches are  \emph{counterfactual regret minimization} (CFR)~\cite{zinkevich2008regret} and CFR+~\cite{tammelin2015solving,tammelin2014solving}, which originated variants such as~\cite{brown2018solving,brown2018deep}.
These techniques were the key for many recent remarkable results~\cite{bowling2015heads,brown2017safe,brown2018superhuman,moravvcik2017deepstack}.
However, these algorithms work with players' behavioral strategies rather than with correlated strategies, and, thus, they are not guaranteed to approach CCEs in general-sum games, even with two players.
The only known theoretical guarantee of CFR when applied to multi-player, general-sum games is that it excludes dominated actions~\cite{gibson2013regret}.
Some works also attempt to apply CFR to multi-player, zero-sum games, see, \emph{e.g.},~\cite{risk2010using}. 

\paragraph{Original contributions} First, we complete the picture on the computational complexity of finding an optimal CCE in sequential games, showing that the problem is inapproximable (\emph{i.e.}, not in Poly-\textsf{APX}) unless \textsf{P} = \textsf{NP} in games with three or more players (chance included).
In the rest of the paper, we focus on how to compute approximate CCEs in multi-player, general-sum, sequential games using no-regret-learning procedures.
We start pointing out simple examples where CFR-like algorithms available in the literature cannot be directly employed to our purpose, as they only provide players' average strategies whose product is not guaranteed to converge to an approximate CCE.
However, we show how CFR can be easily adapted to approach the set of CCEs in multi-player, general-sum sequential games by resorting to sampling procedures (we call the resulting, naive algorithm CFR-S).
Then, we design an enhanced version of CFR (called CFR-Jr) which computes an averaged correlated strategy which is guaranteed the convergence to an approximate CCE with a bound on the regret sub-linear in the size of the game tree.
The key component of CFR-Jr is a polynomial-time algorithm which constructs, at each iteration, the players' normal-form strategies by working on the game tree, avoiding to build the (exponential-sized) normal-form representation. 
We evaluate the scalability of CFR-S and CFR-Jr on standard testbeds. 
While both algorithms solve instances which are orders of magnitude larger than those solved by previous state-of-the-art CCE-finding techniques, CFR-Jr dramatically outperforms CFR-S.
Moreover, CFR-Jr proved to be a good heuristic to compute optimal CCEs, returning nearly-socially-optimal solutions in all the instances of our testbeds.

%% file: content/preliminaries.tex
\section{Preliminaries}\label{sec:preliminaries}

In this section, we introduce some basic concepts which are used in the rest of the paper (see~\citet{shoham2008multiagent} and~\citet{cesa2006prediction} for further details).

\subsection{Extensive-form games and relevant solution concepts}

We focus on \emph{extensive-form games} (EFGs) with imperfect information and perfect recall.
We denote the set of players as $\mathcal{P}\cup\{c\}$, where $c$ is the \emph{Nature} (\emph{chance}) player (representing exogenous stochasticity) selecting actions with a fixed known probability distribution.
$H$ is the set of nodes of the game tree, and a node $h \in H$ is identified by the ordered sequence of actions from the root to the node.
$Z \subseteq H$ is the set of terminal nodes, which are the leaves of the game tree.
For every $h \in H \setminus Z$, we let $P(h)$ be the unique player who acts at $h$ and $A(h)$ be the set of actions she has available.
We write $h \cdot a$ to denote the node reached when $a \in A(h)$ is played at $h$.
For each player $i \in \mathcal{P}$, $u_i: Z \rightarrow \mathbb{R}$ is the payoff function.
We denote by $\Delta$ the maximum range of payoffs in the game, \emph{i.e.}, 
$\Delta=\max_{i\in\pl}\left(\max_{z\in Z}u_i(z)-\min_{z\in Z}u_i(z)\right)$. 
%

We represent imperfect information using \emph{information sets} (from here on, infosets). 
Any infoset $I$ belongs to a unique player $i$, and it groups nodes which are indistinguishable for that player, \emph{i.e.}, $A(h) = A(h')$ for any pair of nodes $h, h' \in I$. 
$\mathcal{I}_i$ denotes the set of all player $i$'s infosets, which form a partition of $\{ h \in H \mid P(h)=i \}$.
We denote by $A(I)$ the set of actions available at infoset $I$.
In perfect-recall games, the infosets are such that no player forgets information once acquired.

We denote with $\pi_i$ a \emph{behavioral strategy} of player $i$, which is a vector defining a probability distribution at each player~$i$'s infoset. 
Given $\pi_i$, we let $\pi_{i,I}$ be the (sub)vector representing the probability distribution at $I\in\I_i$, with $\pi_{i,I,a}$ denoting the probability of choosing action $a\in A(I)$.
%
%
%

An EFG has an equivalent tabular (\emph{normal-form}) representation.
%
A \textit{normal-form plan} for player $i$ is a vector $\sigma_i\in\Sigma_i=\bigtimes_{I\in\mathcal{I}_i} A(I)$ which specifies an action for each player $i$'s infoset.
%
%
Then, an EFG is described through a $|\pl|$-dimensional matrix specifying a utility for each player at each \emph{joint normal-form plan} $\sigma\in\Sigma=\bigtimes_{i\in \pl}\Sigma_i$. 
%
%
The expected payoff of player $i$, when she plays $\sigma_i \in \Sigma_i$ and the opponents play normal-form plans in $\sigma_{-i} \in \Sigma_{-i} = \bigtimes_{j \neq i \in \mathcal{P}} \Sigma_j$, is denoted, with an overload of notation, by $u_i(\sigma_i,\sigma_{-i})$.
Finally, a \textit{normal-form strategy} $x_i$ is a probability distribution over $\Sigma_i$. 
We denote by $\mathcal{X}_i$ the set of the normal-form strategies of player $i$.
Moreover, $\X$ denotes the set of joint probability distributions defined over $\Sigma$.
We also introduce the following notation.
We let $\rho^{\pi_i}$ be a vector in which each component $\rho^{\pi_i}_z$ is the probability of reaching the terminal node $z \in Z$, given that player $i$ adopts the behavioral strategy $\pi_i$ and the other players play so as to reach $z$.
Similarly, given a normal-form plan $\sigma_i \in \Sigma_i$, we define the vector $\rho^{\sigma_i}$.
%
Moreover, with an abuse of notation, $\rho^{\pi_i}_I$ and $\rho^{\sigma_i}_I$ denote the probability of reaching infoset $I \in \I_i$.
Finally, $Z(\sigma_i) \subseteq Z$ is the subset of terminal nodes which are (potentially) reachable if player $i$ plays according to $\sigma_i \in \Sigma_i$.

%
%

The classical notion of CE by~\citet{aumann1974subjectivity} models correlation via the introduction of an external mediator who, before the play, draws the joint normal-form plan $\sigma^\ast\in\Sigma$ according to a publicly known $x^\ast\in\X$, and privately communicates each \emph{recommendation} $\sigma_i^\ast$ to the corresponding player.
After observing their recommended plan, each player decides whether to follow it or not.
A CCE is a relaxation of the CE, defined by~\citet{moulin1978}, which enforces protection against deviations which are independent from the sampled joint normal-form plan.
\begin{definition}\label{def:cce}
	A CCE of an EFG is a probability distribution $x^\ast\in\X$ such that, for every $i\in\pl$, and $\sigma_i'\in\Sigma_i$, it holds:
	\[
	\sum_{\sigma_i\in\Sigma_{i}}\sum_{\sigma_{-i}\in\Sigma_{-i}}x^\ast(\sigma_i,\sigma_{-i})\left(u_i(\sigma_i,\sigma_{-i})-u_i(\sigma_i',\sigma_{-i})\right)\geq 0.
	\]
\end{definition}
CCEs differ from CEs in that a CCE only requires that following the suggested plan is a best response in expectation, before the recommended plan is actually revealed. 
In both equilibrium concepts, the entire probability distribution according to which recommendations are drawn is revealed before the game starts. After that, each player commits to playing a normal-form plan (see Appendix~\ref{sec:appendix_eq} for further details on the various notions of correlated equilibrium in EFGs).
An NE~\cite{nash1951non} is a CCE which can be written as a product of players' normal-form strategies $x_i^\ast \in \X_i$.
In conclusion, an $\varepsilon$-CCE is a relaxation of a CCE in which every player has an incentive to deviate less than or equal to $\varepsilon$ (the same definition holds true for $\varepsilon$-CE and $\varepsilon$-NE).

\subsection{Regret and regret minimization}\label{sec:preliminari_regret}

In the \emph{online convex optimization} framework~\cite{zinkevich2003online}, each player $i$ plays repeatedly against an unknown environment by making a series of decisions $x_i^1,x_i^2,\ldots,x_i^t$.
In the basic setting, the decision space of player $i$ is the whole normal-form strategy space $\X_i$.
At iteration $t$, after selecting $x_i^t$, player $i$ observes a utility $u_i^t(x_i^t)$. 
The \emph{cumulative external regret} of player $i$ up to iteration $T$ is defined as
\begin{equation}\label{eq:regret}
R_i^T=\max_{\hat x_i\in \X_i}\sum_{t=1}^Tu_i^t(\hat x_i)-\sum_{t=1}^Tu_i^t(x_i^t).
\end{equation}

A \emph{regret minimizer} is a function providing the next player $i$'s strategy $x_i^{t+1}$ on the basis of the past history of play and the observed utilities up to iteration $t$. 
A desirable property for regret minimizers is \emph{Hannan consistency}~\cite{hannan1957approximation}, which requires that  $\limsup_{T\to\infty} \frac{1}{T} R_i^T\leq 0$, \emph{i.e.}, the cumulative regret grows at a sublinear rate in the number of iterations $T$.

In an EFG, the regret can be defined at each infoset. 
After $T$ iterations, the cumulative regret for not having selected action $a\in A(I)$ at infoset $I \in \I_i$ (denoted by $R_I^T(a)$) is the cumulative difference in utility that player $i$ would have experienced by selecting $a$ at $I$ instead of following the behavioral strategy $\pi_i^t$ at each iteration $t$ up to $T$.
Then, the regret for player $i$ at infoset $I\in \I_i$ is defined as $R_I^T=\max_{a\in A(I)} R_I^T(a)$.  
Moreover, we let $R_I^{T,+}(a)=\max\{R_I^T(a),0\}$. 

\begin{figure*}[t]
	\centering
	\includegraphics[width=0.9\textwidth]{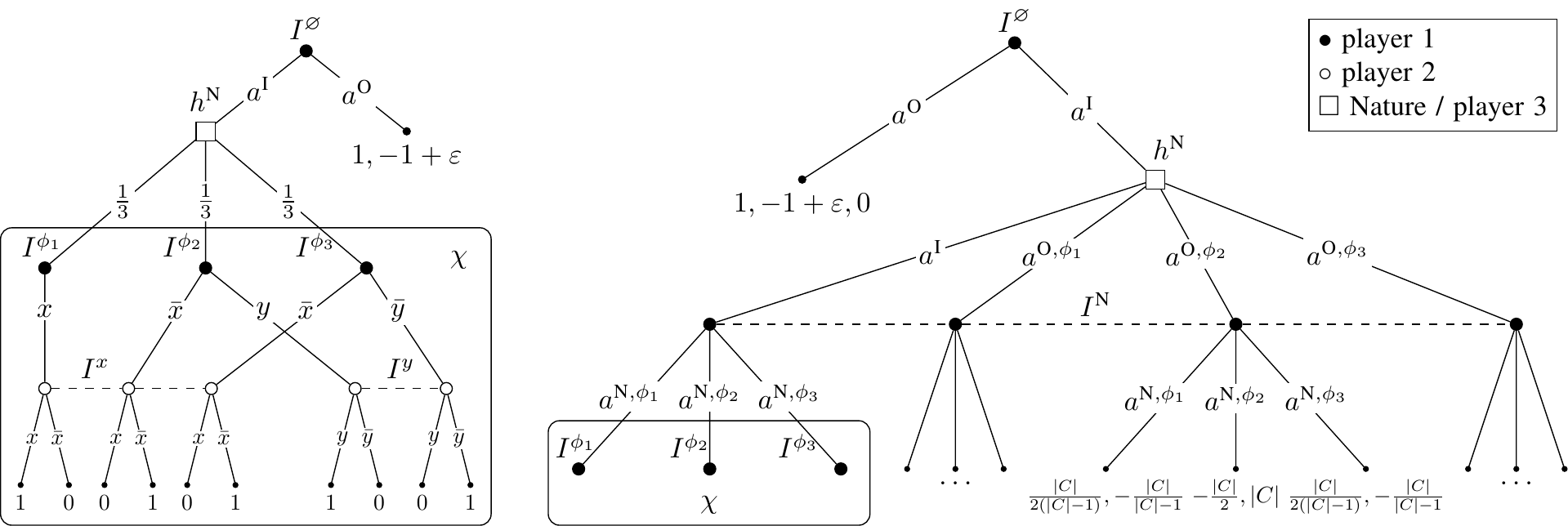}
	\caption{On the left, an example of game for the reduction of Theorem~\ref{thm:apx_hard_nature}, where $V=\{ x,y,z\}$, $C = \{ \phi_1, \phi_2,\phi_3 \}$, $\phi_1 = x$, $\phi_2= \bar x \vee y$, and $\phi_3= \bar x \vee y$. On the right, an example of game for the reduction of Theorem~\ref{thm:apx_hard_no_nature}, with $V$and $C$ as before.}
	\label{fig:reduction}
\end{figure*}

RM~\cite{hart2000simple} is the most widely adopted regret-minimizing scheme when the decision space is $\X_i$ (\emph{e.g.}, in normal-form games).
In the context of EFGs, RM is usually applied locally at each infoset, where the player selects a distribution over available actions proportionally to their positive regret.
Specifically, at iteration $T+1$ player $i$ selects actions $a\in A(I)$ according to the following probability distribution: 
\[
\pi_{i,I,a}^{T+1}=\begin{cases}
\frac{R_{I}^{T,+}(a)}{\sum_{a'\in A(I)}R_{I}^{T,+}(a')}, & \mbox{if } \sum_{a'\in A(I)}R_I^{T,+}(a')	>0\\
\frac{1}{|A(I)|}, & \mbox{otherwise }
\end{cases}.
\]
Playing according to RM at each iteration guarantees, on iteration $T$, $R^T_I\leq \Delta\frac{\sqrt{|A(I)|}}{\sqrt{T}}$~\cite{cesa2006prediction}.
CFR~\cite{zinkevich2008regret} is an anytime algorithm to compute $\varepsilon$-NEs in two-player, zero-sum EFGs.
CFR minimizes the external regret $R_i^T$ by employing RM locally at each infoset.
%
%
In two-player, zero-sum games, if both players have cumulative regrets such that $\frac{1}{T}R_i^T \leq \varepsilon$, then their average behavioral strategies are a $2\varepsilon$-NE~\cite{waugh2009abstraction}.
CFR+ is a variation of classical CFR which exhibits better practical performances~\cite{tammelin2015solving}. 
However, it uses alternation (\emph{i.e.}, it alternates which player updates her regret on each iteration), which complicates the theoretical analysis to prove convergence~\cite{farina2018online,burch18revisiting}.

%% file: content/inapx.tex
\section{Hardness of approximating optimal CCEs}\label{sec:inapx}

We address the following question: given an EFG, can we find a social-welfare-maximizing CCE in polynomial time?
As shown by~\citet{celli2018computing}, the answer is yes in two-player EFGs without Nature.
Here, we give a negative answer to the question in the remaining cases, \emph{i.e.}, two-player EFGs with Nature (Theorem~\ref{thm:apx_hard_nature}) and EFGs with three or more players without Nature (Theorem~\ref{thm:apx_hard_no_nature}).
Specifically, we provide an even stronger negative result: there is no polynomial-time approximation algorithm which finds a CCE whose value approximates that of a social-welfare-maximizing CCE up to any polynomial factor in the input size unless \textsf{P} = \textsf{NP}.
\footnote{Formally, an $r$-approximation algorithm $\mathcal{A}$ for an optimization problem is such that $\frac{\textsc{Opt}}{\textsc{Apx}} \leq r$, where \textsc{Opt} is the value of an optimal solution and \textsc{Apx} is the value of the solution returned by $\mathcal{A}$. See~\cite{ausiello2012complexity} for additional details.}
We prove our results by means of a reduction from \textsc{Sat}, a well known \textsf{NP}-complete problem~\cite{garey1979computers}, which reads as follows.

%
%
%
%

\begin{definition}[\textsc{Sat}]
	Given a finite set $C$ of clauses defined over a finite set $V$ of variables, is there a truth assignment to the variables which satisfies all clauses?
\end{definition}

For clarity, Figure~\ref{fig:reduction} shows concrete examples of the EFGs employed for the reductions of Theorems~\ref{thm:apx_hard_nature}~and~\ref{thm:apx_hard_no_nature}.
Here, we only provide proof sketches, while we report full proofs in Appendix~\ref{sec:appendix_inapx}.

\begin{restatable}{theorem}{thmInapxA}\label{thm:apx_hard_nature}
	Given a two-player EFG with Nature, the problem of computing a social-welfare-maximizing CCE is not in Poly-\textsf{APX} unless \textsf{P} = \textsf{NP}.
	\footnote{Poly-\textsf{APX} is the class of optimization problems that admit a polynomial-time $\mathsf{poly}(\eta)$-approximation algorithm, where $\mathsf{poly}(\eta)$ is a polynomial function of the input size $\eta$~\cite{ausiello2012complexity}.}
\end{restatable}

\begin{proof}[Proof sketch]
	An example of our reduction from \textsc{Sat} is provided on the left of Figure~\ref{fig:reduction}.
	Its main idea is the following: player 2 selects a truth assignment to the variables, while player 1 chooses a literal for each clause in order to satisfy it.
	It can be proved that there exists a CCE in which each player gets utility $1$ if and only if \textsc{Sat} is satisfiable (as player 1 selects $a^\textsc{I}$), otherwise player 1 plays $a^\textsc{O}$ in any CCE and its social welfare is $\varepsilon$.
	Assume there is a a polynomial-time $\mathsf{poly}(\eta)$-approximation algorithm $\mathcal{A}$.
	If \textsc{Sat} is satisfiable, $\mathcal{A}$ would return a CCE with social welfare at least $ \frac{2}{\mathsf{poly}(\eta)}$.
	Since, for $\eta$ sufficiently large it holds $\frac{2}{\mathsf{poly}(\eta)} > \frac{1}{2^{\eta}}$, then $\mathcal{A}$ would allow us to decide in polynomial time whether \textsc{Sat} is satisfiable, leading to a contradiction unless \textsf{P} = \textsf{NP}.
\end{proof}

\begin{restatable}{theorem}{thmInapxB}\label{thm:apx_hard_no_nature}
	Given a three-player EFG without Nature, the problem of computing a social-welfare-maximizing CCE is not in Poly-\textsf{APX} unless \textsf{P} = \textsf{NP}.
\end{restatable}

\begin{proof}[Proof sketch]
	An example of our reduction from \textsc{Sat} is provided on the right of Figure~\ref{fig:reduction}.
	It is based on the same idea as that of the previous proof, where the uniform probability distribution played by Nature is simulated by a particular game gadget (requiring a third player). 
\end{proof}

%% file: content/cfr_s.tex
\section{CFR in multi-player general-sum sequential games}\label{sec:cfr_s}

In this section, we first highlight why CFR cannot be directly employed when computing CCEs of general-sum games.
Then, we show a simple way to amend it.

\subsection{Convergence to CCEs in general-sum games}

When players follow strategies recommended by a regret minimizer, the \emph{empirical frequency of play} approaches the set of CCEs~\cite{cesa2006prediction}.
Suppose that, at time $t$, the players play a joint normal-form plan $\sigma^t \in \Sigma$ drawn according to their current strategies.
Then, the empirical frequency of play after $T$ iterations is defined as the joint probability distribution $\bar x^T \in \mathcal{X}$ such that $\bar x^T(\sigma) \coloneqq \frac{| t \leq T : \sigma^t = \sigma |}{T}$ for every $\sigma \in \Sigma$.
However, vanilla CFR and its most popular variations (such as, \emph{e.g.}, CFR+~\cite{tammelin2015solving} and DCFR~\cite{brown2018solving}) do not keep track of the empirical frequency of play, as they only keep track of the players' average behavioral strategies. 
This ensures that the strategies are compactly represented, but it is not sufficient to recover a CCE in multi-player, general-sum games. 
Indeed, it is possible to show that, even in normal-form games, if the players play according to some regret-minimizing strategies, then the product distribution $x \in \mathcal{X}$ resulting from players' (marginal) average strategies may not converge to a CCE.
In order to see this, we provide the following simple example.

\begin{figure}
	\centering
	\renewcommand{\arraystretch}{1.2}
	\begin{tabular}{r|c|c|}
			\multicolumn{1}{r}{}
			&  \multicolumn{1}{c}{$ \sigma_{L} $}
			& \multicolumn{1}{c}{$ \sigma_{R} $} \\
			\cline{2-3}
			$ \sigma_{L} $ & $1,1$ & $1,0$ \\
			\cline{2-3}
			$ \sigma_{R} $ & $0,1$ & $1,1$ \\
			\cline{2-3}
	\end{tabular}
	\caption{Game where $\bar x_1^T\otimes \bar x_2^T$ does not converge to a CCE.}
	\label{fig:game_es}
\end{figure}

\paragraph{Example}
Consider the two-player normal-form game depicted in Figure~\ref{fig:game_es}.
At iteration $t$, let players' strategies $x_1^t, x_2^t$ be such that $x_1^t (\sigma_{L})= x_2^t(\sigma_{L}) = (t+1)\mod 2$.
Clearly, $u_1^t(x^t) = u_2^t(x^t) = 1$ for any $t$.
For both players, at iteration $t$, the regret of not having played $\sigma_{L}$ is $0$, while the regret of $\sigma_{R}$ is $-1$ if and only if $t$ is even, otherwise it is $0$.
As a result, after $T$ iterations, $R_1^T = R_2^T = 0$, and, thus, $x_1^t$ and $x_2^t$ minimize the cumulative external regret.
Players' average strategies  $\bar x_1^T=\frac{1}{T} \sum_{t=1}^{T} x_1^t$ and $\bar x_2^T=\frac{1}{T} \sum_{t=1}^{T} x_2^t$ converge to $(\frac{1}{2},\frac{1}{2})$ as $T \rightarrow \infty$.
However, $x \in \mathcal{X}$ such that $x(\sigma) = \frac{1}{4}$ for every $\sigma \in \Sigma$ is not a CCE of the game. Indeed, a player is always better off playing $\sigma_{L}$, obtaining a utility of $1$, while she only gets $\frac{3}{4}$ if she chooses to stick to $x$.
We remark that $\bar x^T$ converges, as $T \rightarrow \infty$, to $x \in \mathcal{X}:x(\sigma_{L},\sigma_{L}) = x(\sigma_{R},\sigma_{R}) = \frac{1}{2}$, which is a CCE.

\begin{figure}
	\centering
	\includegraphics[width=0.45\textwidth]{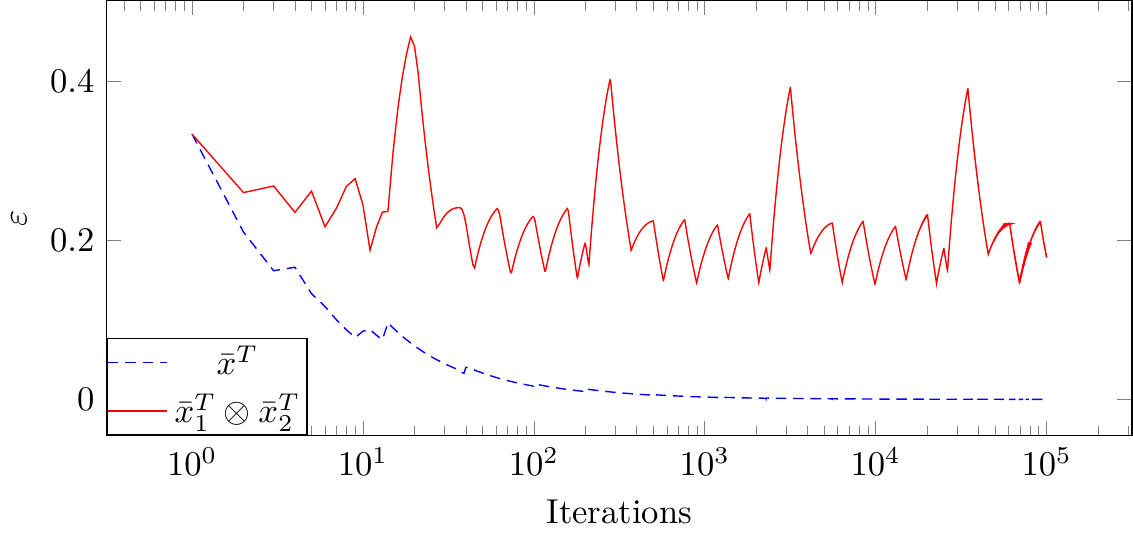}
	\caption{Approximation attained by $\bar x^T$ and $\bar x^T_1\otimes\bar x^T_2$.}
	\label{fig:esempi}
\end{figure}

The example above employs handpicked regret-minimizing strategies, but similar examples can be easily found when applying common regret minimizers.
As an illustrative case, Figure~\ref{fig:esempi} shows, that, even with a simple variation of the Shapley game (see Appendix~\ref{sec:appendix-cfr-s}), the outer product of the average strategies $\bar x_1^T\otimes \bar x_2^T$ obtained via RM does not converge to a CCE as $T\to\infty$.
It is clear that the same issue may (and does, see Figures \ref{fig:goof} and \ref{fig:shapley_game}) happen when directly applying CFR to general-sum EFGs.
%
%

\subsection{CFR with sampling (CFR-S)}\label{sec:sub_cfr_s}

\begin{figure}
	\begin{algorithm}[H]
		\centering
		\scriptsize
		\caption{\texttt{CFR-S for player $i$}}
		\begin{algorithmic}[1]
			\Function{\textsc{CFR-S}}{$\Gamma$,$i$}
			\State Initialize a regret minimizer for each $I\in\I_i$
			\State $t\gets 1$
			\While{$t < T$}
			\State $\sigma_i^t\gets\textsc{Recommend}(I_\varnothing)$
			\State Observe $u_i^t(\sigma_i)\coloneqq u_i(\sigma_i,\sigma_{-i}^t)$\label{line:observe}
			\State $\textsc{Update}(I_\varnothing,\sigma_i^t,u_i^t)$
			\State $t\gets t+1$
			\EndWhile
			\EndFunction
		\end{algorithmic}
		\label{alg:cfr-s}
	\end{algorithm}
\end{figure}

Motivated by the previous examples, we describe a simple variation of CFR guaranteeing approachability to the set of CCEs even in multi-player, general-sum EFGs.
Vanilla CFR proceeds as follows (see Subsection~\ref{sec:preliminari_regret} for the details): for each iteration $t$, and for each infoset $I\in\I_i$, player $i$ observes the realized utility for each action $ a\in A(I)$, and then computes $\pi_{i,I}^t$ according to standard RM. 
Once $\pi_{i,I}^t$ has been computed, it is used by the regret minimizers of infosets on the path from the root to $I$ so as to compute observed utilities.
We propose CFR \emph{with sampling} (CFR-S) as a simple way to keep track of the empirical frequency of play. 
The basic idea is letting each player~$i$, at each $t$, draw $\sigma_i^t$ according to her current strategy.
Algorithm~\ref{alg:cfr-s} describes the structure of CFR-S, where function \textsc{Recommend} builds a normal-form plan $\sigma_i^t$ by sampling, at each $I\in\I_i$, an action in $A(I)$ according to $\pi_i^t$ computed via RM, and \textsc{Update} updates the average regrets local to each regret minimizer by propagating utilities according to $\sigma_i^t$.
Each player $i$ experiences utilities depending, at each $t$, on the sampled plans $\sigma_{-i}^t$ (Line~\ref{line:observe}).
Joint normal form plans $\sigma^t\coloneqq (\sigma_i^t,\sigma_{-i}^t)$ can be easily stored to compute the empirical frequency of play.
We state the following (see Appendix~\ref{sec:appendix-cfr-s} for detailed proofs):
\begin{restatable}{theorem}{cfrs}\label{th:cfr-s}
	The empirical frequency of play $\bar x^T$ 
	obtained with CFR-S converges to a CCE almost surely, for $T\to\infty$.
\end{restatable}
Moreover, the cumulative regret grows as $O(T^{-1/2})$.
This result is in line with the approach of \citet{hart2000simple} in normal-form games.
Despite its simplicity, we show (see Section~\ref{sec:exp_eval} for an experimental evaluation) that it is possible to achieve better perfomances via a smarter reconstruction technique that keeps CFR deterministic, avoiding any sampling step.

%% file: content/cfr_jr.tex
\section{CFR with joint distribution reconstruction (CFR-Jr)}\label{sec:cfr_jr}

We design a new method---called \emph{CFR with joint distribution reconstruction} (CFR-Jr)---to enhance CFR so as to approach the set of CCEs in multi-player, general-sum EFGs.
Differently from the naive CFR-S algorithm, CFR-Jr does not sample normal-form plans, thus avoiding any stochasticity.


The main idea behind CFR-Jr is to keep track of the average joint probability distribution $\bar x^T \in \X$ arising from the regret-minimizing strategies built with CFR.
Formally, $\bar x^T = \frac{1}{T}\sum_{t = 1}^{T} x^t$, where $x^t \in \X$ is the joint probability distribution defined as the product of the players' normal-form strategies at iteration $t$.
At each $t$, CFR-Jr computes $\pi_{i}^t$ with CFR's update rules, and then constructs a strategy $x_i^t \in \X_i$ which is realization equivalent (\emph{i.e.}, it induces the same probability distribution on the terminal nodes, see~\cite{shoham2008multiagent} for a formal definition) to $\pi_i^t$.
We do this efficiently by directly working on the game tree, without resorting to the normal-form representation.
Strategies $x_i^t$ are then employed to compute $x^t$.
The pseudocode of CFR-Jr is provided in Appendix~\ref{appendix:cfr-jr}.

\begin{figure}
	\begin{algorithm}[H]
		\centering
		\scriptsize
		\caption{\texttt{Reconstruct $x_i$ from $\pi_i$}}
		\begin{algorithmic}[1]
			\Function{\textsc{NF-Strategy-Reconstruction}}{$\pi_i$}
			\State $\mathbf{X} \gets \varnothing$\Comment{$\mathbf{X}$ is a dictionary defining $x_i$}
			\State $\omega_z \gets \rho^{\pi_i}_z \,\,\,\, \forall z \in Z $\label{line:recon_init}
			\While{$\omega > 0$}
			\State $\bar {\sigma}_i \gets \argmax_{\sigma_i \in \Sigma_i} \min_{z \in Z(\sigma_i)} \omega_z$\label{line:recon_maxmin}
			\State $\bar{\omega} \gets \min_{z \in Z(\bar{\sigma}_i)} \omega_i(z)$
			\State $\mathbf{X} \gets \mathbf{X} \cup ( \bar \sigma_i, \bar \omega )$
			\State $\omega \gets \omega - \bar{\omega} \, \rho^{\bar{\sigma}_i}$
			\EndWhile
			\Return $x_i$ built from the pairs in $\mathbf{X}$
			\EndFunction
		\end{algorithmic}
		\label{alg:algo_recon}
	\end{algorithm}
\end{figure}

Algorithm~\ref{alg:algo_recon} shows a polynomial-time procedure to compute a normal-form strategy $x_i \in \X_i$ realization equivalent to a given behavioral strategy $\pi_i$.
The algorithm maintains a vector ${\omega}$ which is initialized with the probabilities of reaching the terminal nodes by playing $\pi_i$ 
(Line~\ref{line:recon_init}), and it works by iteratively assigning probability to normal-form plans so as to induce the same distribution of $\omega$ over $Z$.
\footnote{Vector $\omega$ is a \emph{realization-form strategy}, as defined by~\citet[Definition 2]{farina2018exante}.}
In order for this to work, at each iteration, the algorithm must pick a normal-form plan $\bar \sigma_i \in \Sigma_i$ which maximizes the minimum (remaining) probability $\omega_z$ over the terminal nodes $z \in Z(\bar \sigma_i)$ reachable when playing $\bar \sigma_i$ (Line~\ref{line:recon_maxmin}).
Then, the probabilities $\omega_z$ for $z \in Z(\bar \sigma_i)$ are decreased by the minimum (remaining) probability $\bar \omega$ corresponding to $\bar \sigma_i$, and $\bar \sigma_i$ is assigned probability $\bar \omega$ in $x_i$.
The algorithm terminates when the vector $\omega$ is zeroed, returning a normal-form strategy $x_i$ realization equivalent to $\pi_i$.
This is formally stated by the following result, which also provides a polynomial (in the size of the game tree) upper bound on the running time of the algorithm and on the support size of the returned normal-form strategy $x_i$.~\footnote{Given a normal-form strategy $x_i \in \X_i$, its \emph{support} is defined as the set of $\sigma_i \in \Sigma_i$ such that $x_i(\sigma_i) > 0$.}

\begin{restatable}{theorem}{thmAlgoRecon}\label{thm:corr_algo_rec}
	Algorithm~\ref{alg:algo_recon} outputs a normal-form strategy $x_i \in \X_i$ realization equivalent to a given behavioral strategy $\pi_i$, and it runs in time $O(|Z|^2)$. Moreover, $x_i$ has support size at most $|Z|$.
\end{restatable}

Intuitively, the result in Theorem~\ref{thm:corr_algo_rec} (its full proof is in Appendix~\ref{appendix:cfr-jr}) relies on the crucial observation that, at each iteration, there is at least one terminal node $z \in Z$ whose corresponding probability $\omega_z$ is zeroed during that iteration. 
The algorithm is guaranteed to terminate since each $\omega_z$ is never negative,
%
which is the case given how the normal-form plans are selected (Line~\ref{line:recon_maxmin}), 
and since the game has perfect recall.
This guarantees that the algorithm eventually terminates in at most $|Z|$ iterations.

Finally, the following theorem (whose full proof is in Appendix~\ref{appendix:cfr-jr}) proves that the average distribution $\bar x^T$ obtained with CFR-Jr approaches the set of CCEs.
Formally:

\begin{restatable}{theorem}{thmCFRJr}\label{thm:eps_cce_avg_prod}
	If $\frac{1}{T}R_i^T \leq \varepsilon$ for each player $i \in \mathcal{P}$, then $\bar x^T$ obtained with CFR-Jr is an $\varepsilon$-CCE.
\end{restatable}

\begin{figure*}[t]
	\centering
	\includegraphics[width=.98\linewidth]{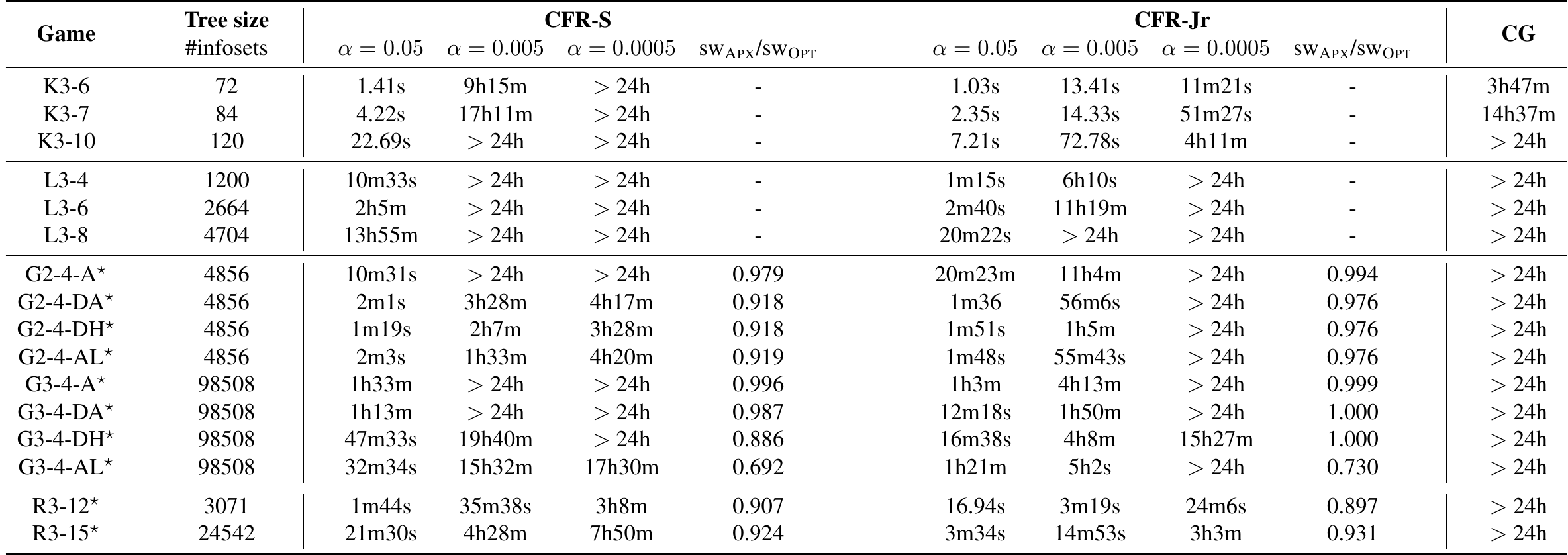}
	\caption{Comparison between the run time and the social welfare of CFR-S, CFR-Jr (for various levels of accuracy $\alpha$), and the CG algorithm. General-sum instances are marked with $^\star$. Results of CFR-S are averaged over 50 runs. We generated 20 instances for each R$p$-$d$ family.}
	\label{tab:convergence}
\end{figure*}

%
%
This is a direct consequence of the connection between regret-minimizing procedures and CCEs, and of the fact that $\bar x^T$ is obtained by averaging the products of normal-form strategies which are equivalent to regret-minimizing behavioral strategies obtained with CFR.
%

%% file: content/experimental.tex
\section{Experimental evaluation}\label{sec:exp_eval}

We experimentally evaluate CFR-Jr, comparing its performance with that of CFR-S, CFR, and the  state-of-the-art algorithm for computing optimal CCEs (denoted by CG)~\cite{celli2018computing}.\footnote{
	The only other known algorithm to compute a CCE is by~\citet{huang2008computing} (see also~\cite{jiang2015polynomial} for an amended version). 
	However, this algorithm relies on the ellipsoid method, which is inefficient in practice~\cite{grotschel1981ellipsoid}.
}
This algorithm is a variation of the simplex method employing a column generation technique based on a MILP pricing oracle (we use the GUROBI 8.0 MILP solver).
Notice that directly applying RM on the normal form is not feasible, as $|\Sigma| > 10^{20}$ even for the smallest instances.
%
%
Further results are in Appendix~\ref{sec:appendix_exp}.

\begin{figure}
	\begin{minipage}{.20\textwidth}
		\centering
		\includegraphics[width=1.1\textwidth]{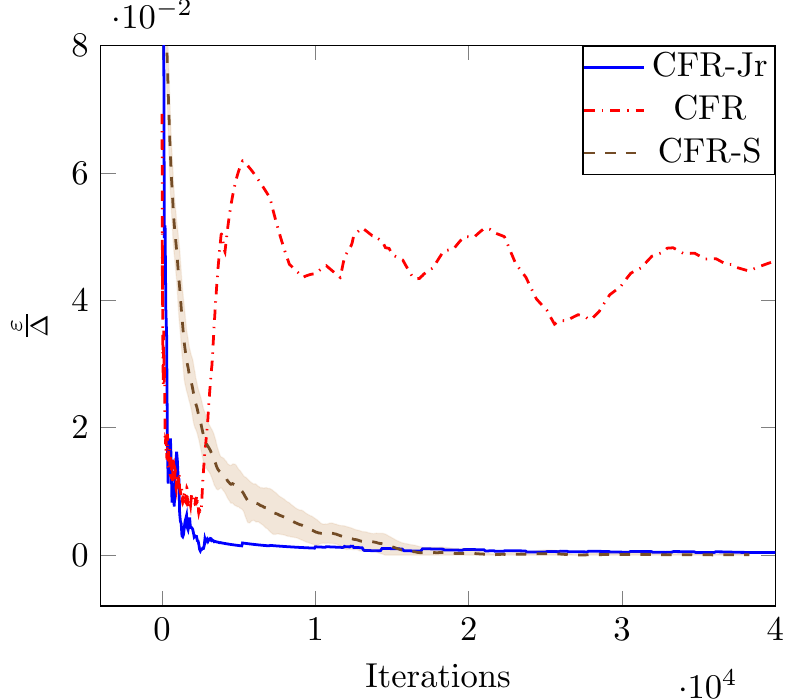}
	\end{minipage}
		\hspace{0.3cm}
	\begin{minipage}{.20\textwidth}
		\vspace{.2cm}
		\centering
		\includegraphics[width=1.2\textwidth]{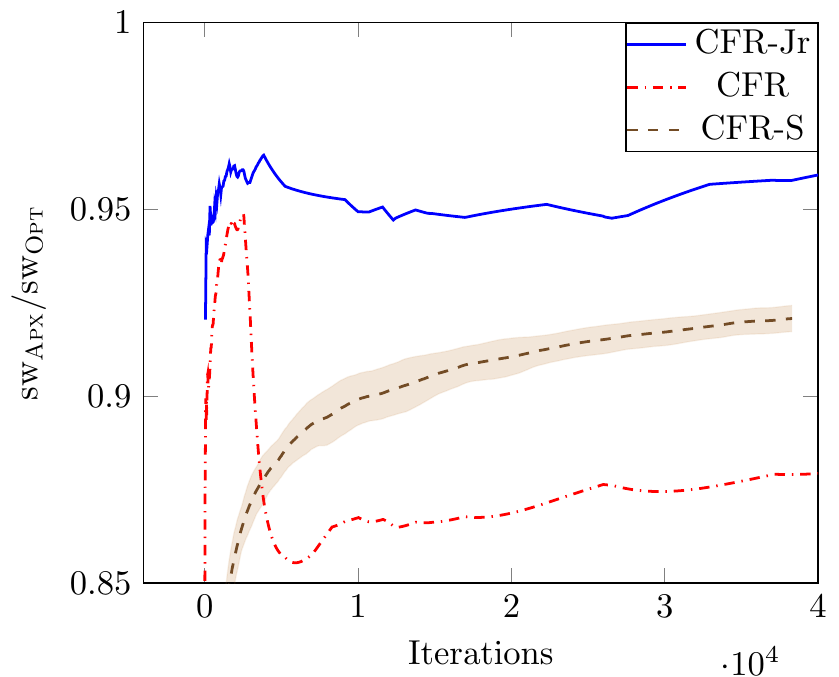}
	\end{minipage}
	\caption{\emph{Left}: Convergence rate attained in G2-4-DA. \emph{Right}: Social welfare attained in G2-4-DA.}
	\label{fig:goof}
\end{figure}

\paragraph{Setup}
We conduct experiments on parametric instances of 3-player Kuhn poker games~\cite{kuhn1950simplified},  3-player Leduc hold'em poker games~\cite{southey2005bayes}, 2/3-player Goofspiel games~\cite{ross1971goofspiel}, and some randomly generated general-sum EFGs.
The 2-player zero-sum versions of these games are standard benchmarks for imperfect-information game solving.
%
In Appendix~\ref{sec:appendix_exp}, we describe their multi-player, general-sum counterparts.
Each instance is identified by parameters $p$ and $r$, which denote, respectively, the number of players and the number of ranks in the deck of cards. 
For example, a 3-player Kuhn game with rank 4 is denoted by Kuhn3-4, or K3-4.
We use different tie-breaking rules for the Goofspiel instances (denoted by A, DA, DH, AL---see Appendix~\ref{sec:appendix_exp}). 
%
%
R$p$-$d$ denotes a random game with $p$ players, and depth of the game tree $d$.

%
%
\paragraph{Convergence}
We evaluate the run time required by the algorithms to find an approximate CCE. The results are provided in Table~\ref{tab:convergence}, which reports
the run time needed by CFR-S, CFR-Jr to achieve solutions with different levels of accuracy, and the time needed by CG for reaching an equilibrium.\footnote{Table~\ref{tab:convergence} only accounts for algorithms with guaranteed convergence to a CCE (recall that CFR is not guaranteed to converge in multi-player, general-sum EFGs). The original version of the CG algorithm computes an optimal CCE. For our tests, we modified it to stop when a feasible solution is reached.}
The accuracy $\alpha$ of the $\varepsilon$-CCEs reached is defined as $\alpha= \frac{\varepsilon}{\Delta}$.
CFR-Jr consistently outperforms both CFR-S and CG, being orders of magnitude faster.
Figure~\ref{fig:goof}, on the left, shows the performance of CFR-Jr, CFR-S (mean plus/minus standard deviation), and CFR over G2-4-DA in terms of $\varepsilon/\Delta$ approximation.
CFR performs dramatically worse than CFR-S and CFR-Jr.

\paragraph{Social welfare}
Table~\ref{tab:convergence} shows, for the general-sum games,  the social welfare approximation ratio between the social welfare of the solutions returned by the algorithms (sw$_{\textsc{Apx}}$) and the optimal social welfare (sw$_{\textsc{Opt}}$).
The social welfare guaranteed by CFR-Jr is always nearly optimal, which makes it a good heuristic to compute optimal CCEs. 
Reaching a \emph{socially good} equilibrium is crucial, in practice, to make correlation credible.
Figure~\ref{fig:goof}, on the right, shows the performance of CFR-Jr, CFR-S (mean plus/minus standard deviation), and CFR over G2-4-DA in terms of social welfare approximation ratio.
Also in this case, CFR performs worse than the other two algorithms.

%% file: content/discussion.tex
\section{Conclusions and future works}\label{sec:discussion}

In this paper, we proved that finding an optimal (\emph{i.e.}, social-welfare maximizing) CCE is not in Poly-\textsf{APX}, unless \textsf{P} = \textsf{NP}, in general-sum EFGs with two players and chance or with multiple players.
We proposed CFR-Jr as an appealing remedy to the conundrum of computing correlated strategies for multi-player, general-sum settings with game instances beyond toy-problems.
In the future, it would be interesting to further study how to approximate CCEs in other classes of structured games such as, \emph{e.g.}, polymatrix games and congestion games.
Moreover, a CCE strategy profile could be employed as a starting point to approximate tighter solution concepts that admit some form of correlation. 
This could be the case, \emph{e.g.}, of the TMECor~\cite{farina2018exante}, which is used to model collusive behaviors and interactions involving teams.
Finally, it would be interesting to further investigate whether it is possible to define regret-minimizing procedures for general EFGs leading to refinements of the CCEs, such as CEs and EFCEs. 
This begets new challenging problems in the study of how to minimize regret in structured games.

%% file: content/appendix.tex
\appendix

\input{content/appendix_correlati}
\input{content/appendix_inapx}

\input{content/appendix_cfr_s.tex}

\input{content/appendix_cfr_jr.tex}

\input{content/appendix_exp}

%% file: content/appendix_correlati.tex
\section{Discussion of Other Solution Concepts Based on Correlation}\label{sec:appendix_eq}
The classical notion of correlated equilibrium is the one introduced by~\citet{aumann1974subjectivity} for the normal form. Its definition for EFGs is as follows, and it employs the equivalent normal form of the game.
\begin{definition}\label{def:ce}
	A \emph{correlated equilibrium (CE)} of an EFG is a probability distribution $x^\ast\in\X$ such that, for every $i\in \pl$, and for every $\sigma_i,\sigma_i'\in \Sigma_i$, it holds:
	\[
	\sum_{\sigma_{-i}\in\Sigma_{-i}}x^\ast(\sigma_i,\sigma_{-i})\left(u_i(\sigma_i,\sigma_{-i})-u_i(\sigma_i',\sigma_{-i})\right)\geq 0.
	\]
\end{definition}
A CE can be interpreted in terms of a mediator who, \emph{ex ante} the play, draws the joint normal-form plan $\sigma^\ast\in\Sigma$ according to a publicly known $x^\ast\in\X$, and privately communicates each \emph{recommendation} $\sigma_i^\ast$ to the corresponding player.
After observing their recommended plan, each player decides whether to follow it or not. 

CCEs (Definition~\ref{def:cce}) differ from CEs in that a CCE only requires that following the suggested plan is a best response in expectation, before the recommended plan is actually revealed. 
In both CE and CCE, the entire vector of recommendations $\sigma^\ast$, specifying a move for each infoset, is computed before the playing phase of the game (as opposed to other solution concepts involving communication~\cite{forges1993,myerson1986}).
%

Specifically for EFGs, \citet{von2008extensive} introduced the notion of \emph{extensive-form correlated equilibrium} (EFCE). In this solution concept, each recommended action is assumed to be in a sealed envelope and is revealed only when the player reaches the relevant infoset (\emph{i.e}., the infoset where she can make that move).
Therefore, EFCEs require recommendations to be delivered during game execution, which makes them more demanding in terms of communication requirements than CEs and CCEs.
The size of the signal that has to be sampled is the same in all the three solution concepts, and it has polynomial size (one action for each infoset).
The following relation holds between the sets of equilibria described above: CE $\subseteq$ EFCE $\subseteq$ CCE, see~\cite{von2008extensive} for further details.

%% file: content/appendix_inapx.tex
\section{Omitted Proofs for Inapproximability Results}\label{sec:appendix_inapx}

\thmInapxA*

\begin{proof}
	We provide a reduction from \textsc{Sat}.
	Given a \textsc{Sat} instance $(C,V)$, we build a two-player EFG with Nature $\Gamma_\varepsilon(C,V)$ with the following structure:
	\begin{itemize}[nolistsep,itemsep=0mm]
		\item The game starts in $h^\varnothing \in I^\varnothing$, where player 1 chooses an action between $a^\textsc{I}$ and $a^\textsc{O}$. In the first case, the game goes on with $h^\varnothing \cdot a^\textsc{I}=h^\textsc{N}$. Otherwise, the game ends with $u_1(z) = 1$ and $u_2(z) = -1 + \varepsilon$.
		\item At state $h^\textsc{N}$, Nature selects an action among $\{ a^{\phi} \mid \phi \in C \}$ uniformly at random, with $ h^\textsc{N} \cdot a^{\phi} = h^{\phi}$.
		\item Each state $h^{\phi}$ constitutes a player 1's infoset $I^{\phi}$. 
		At $I^{\phi}$, player 1 chooses an action in $\{ a^{\phi,l} \mid l \in \phi \}$, where $l$ denotes a literal in $\phi$. 
		Then, $h^{\phi} \cdot a^{\phi,l} = h^{\phi, l}$. 
		\item All states $h^{\phi, l}$ such that $l = v$ or $l = \bar v$ for some $v \in V$ belong to the same player 2's infoset $I^v$. At $I^v$, player 2 has two actions available, namely $a^{v}$ and $a^{\bar v}$.
		\item Then, the game ends and players' payoffs $u_1(z)= u_2(z)$ are equal to 1 if and only if $z = h^{\phi,l} \cdot a^{l}$, while they are 0 otherwise. 
	\end{itemize}
	
	Intuitively, each of the $2^{|V|}$ player 2's plans corresponds to a truth assignment $\tau$ where variable $v \in V$ is set to \textsc{true} (resp., \textsc{false}) if $a^v$ (resp., $a^{\bar v}$) is played at $I^v$. 
	Moreover, a player 1's plan determines whether the game is played ($a^\textsc{I}$) or not ($a^\textsc{O}$) and, in the first case, it selects one literal for each clause $\phi \in C$ (corresponding to the action played at infoset $I^\phi$).
	%
	If player 1 plays $a^\textsc{I}$, Nature chooses a clause $\phi \in C$ uniformly at random, and, then, the players' payoffs are 1 if and only if player 1 selected a literal of $\phi$ evaluating to \textsc{true} under $\tau$.
	Thus, in this case, players' expected payoffs are equal to the number of literals selected by player 1 evaluating to \textsc{true} under $\tau$, divided by the number of clauses $|C|$.
	As a result, if \textsc{Sat} is satisfiable, then there exists a joint plan where players' expected payoffs are equal to 1.
	It is sufficient that player 2 plays the plan associated to a satisfying truth assignment $\tau$, while player 2 selects a literal evaluating to \textsc{true} under $\tau$ for each clause.
	This is also a CCE with maximum social welfare equal to 2, as it provides the players with their maximum expected payoffs.
	Instead, if \textsc{Sat} is not satisfiable, then any CCE must recommend player 1 to play $a^\textsc{O}$ at $I^\varnothing$, otherwise her expected payoff would be strictly less than 1 and she would have an incentive to deviate to action $a^\textsc{O}$, reaching a payoff of 1.
	Hence, in this case, any CCE has social welfare $\varepsilon$.
	Now, let $\varepsilon = \frac{1}{2^{\eta}}$, where $\eta$ is the size of the \textsc{Sat} instance ($\varepsilon$ can be encoded with a number of bits polynomial in $|C|$ and $|V|$).
	Assume there is a polynomial-time $\mathsf{poly}(\eta)$-approximation algorithm $\mathcal{A}$.
	If \textsc{Sat} is satisfiable, $\mathcal{A}$ applied to $\Gamma_\varepsilon(C,V)$ would return a CCE with social welfare at least $ \frac{2}{\mathsf{poly}(\eta)}$.
	Since, for $\eta$ sufficiently large, $\frac{2}{\mathsf{poly}(\eta)} > \frac{1}{2^{\eta}}$, $\mathcal{A}$ would allow us to decide in polynomial time whether \textsc{Sat} is satisfiable, a contradiction unless \textsf{P} = \textsf{NP}.
\end{proof}

\thmInapxB*

\begin{proof}
	We use a reduction similar to that in Theorem~\ref{thm:apx_hard_nature}. We build a three-player EFG $\hat\Gamma_\varepsilon(C,V)$ such that:
	\begin{itemize}[nolistsep,itemsep=0mm]
		\item The game starts in $h^\varnothing \in I^\varnothing$, as $\Gamma(C,V)$.
		\item At state $h^\textsc{N} $, player 3 plays an action $\{ a^{\textsc{O},\phi} \mid \phi \in C \} \cup \{ a^\textsc{I} \}$, with $h^\textsc{N} \cdot a^{\textsc{I}} = h^{\textsc{I}}$, $h^\textsc{N} \cdot a^{\textsc{O},\phi} = h^{\textsc{O}, \phi}$.
		\item All states $h^{\textsc{O}, \phi}$ and $h^{\textsc{I}}$ belong to a player 1's infoset $I^\textsc{N}$, where she selects an action among $\{  a^{\textsc{N},\phi} \mid \phi \in C \}$.
		\item Then, if player 3 played $a^\textsc{I}$, $h^{\textsc{I}} \cdot a^{\textsc{N},\phi} = h^\phi$ and the games goes on as $\Gamma(C,V)$ (with player 3's payoffs set to zero). Instead, if player 3 played an action $a^{\textsc{O},\phi}$, the game ends with $2 u_1(z) = 2 u_2(z) = - u_3(z) = \frac{ |C|}{|C|-1}$ if $z = h^{\textsc{O},\phi} \cdot a^{\textsc{N},\phi'})$ and $\phi \neq \phi'$, while $2 u_1(z) = 2 u_2(z) = - u_3(z) = -|C|$ if $\phi = \phi'$.
	\end{itemize}
	Intuitively, the introduction of a third player allows us to simulate the random move of Nature in $\Gamma(C,V)$, since, in any CCE of $\hat\Gamma(C,V)$, player 1 is recommended to play a uniform distribution at infoset $I^\textsc{N}$ and player 3 is always told to play action $a^{\textsc{I}}$.
	First, if player 3 is recommended an action $a^{\textsc{O},\phi}$ with positive probability, then player 2 would have an incentive to switch to action $a^\textsc{O}$ at $I^\varnothing$. 
	Moreover, assuming player 3 is told to play $a^\textsc{I}$, if player 1 is recommended to play some action $a^{\textsc{N},\phi}$ with probability $p > \frac{1}{|C|}$, then player 3 would have an incentive to switch to action $a^{\textsc{O},\phi}$, as she would get $p |C| - (1-p) \frac{|C|}{|C|-1} > 0$, while she gets 0 by playing $a^\textsc{I}$.
	Finally, a reasoning similar to that for Theorem~\ref{thm:apx_hard_nature} concludes the proof.
\end{proof}

%% file: content/appendix_cfr_s.tex
\section{CFR-S}\label{sec:appendix-cfr-s}

\subsection{Example}

Figure~\ref{fig:shapley} reports the game employed in the experiments of Figure~\ref{fig:esempi}, where the outer product of the average strategies $\bar x_1^T\otimes \bar x_2^T$ obtained by RM does not converge to a CCE as $T\to\infty$
(for $\bar x_1^T\otimes \bar x_2^T$, $\varepsilon$ of the $\varepsilon$-CCE has a cyclic behavior and does not converge to zero).
\subsection{Omitted Proofs}

The theoretical guarantees of CFR-S can be derived via the framework of~\citet{farina2018online}, as discussed in the following.

At each iteration $t$, let $\sigma_i^t \in \Sigma_i$ be the normal-form plan sampled by player $i$ and $\sigma_{-i}^t \in \Sigma_{-i}$ be the plans drawn by the other players. 
The utility experienced by player~$i$ at stage~$t$ is denoted by $u_i^t(\sigma_i^t)\coloneqq u_i(\sigma_i^t,\sigma_{-i}^t)$. 
The players' observations in CFR-S call for a slight variation in the definition of cumulative regret.
\begin{figure}
\centering
{\renewcommand{\arraystretch}{1.1}
	\begin{tabular}{|c|c|c|}
		\hline
		1,0 & 0,1 & 0,0\\
		\hline
		0,0 & 2,0 & 0,1\\
		\hline
		0,1 & 0,0 & 1,0\\
		\hline
\end{tabular}\vspace{.3cm}}
\caption{A simple variation of the Shapley game.}
\label{fig:shapley}
\end{figure} 
After $T$ iterations, we define the cumulative regret experienced by player $i$ as
\begin{equation}\label{eq:sampled_regret}
\tilde R_i^T\coloneqq \max_{\hat\sigma_i\in \Sigma_i}\sum_{t=1}^T\left( u_i^t(\hat \sigma_i)- u_i^t(\sigma_i^t)\right).
\end{equation}
The connection between the cumulative regret and the set of CCEs remains unchanged when the regret is defined as in Equation~\eqref{eq:sampled_regret}, as shown by the following result (whose proof is similar to that of~\cite{hart2000simple}).
\begin{theorem}\label{thm:convergence-to-cce}
	If $\lim\sup_{T \to \infty} \frac{1}{T} \tilde R_i^T \leq 0$ almost surely for each player $i\in\pl$, then the empirical frequency of play $\bar x^T$ converges almost surely as $T \to \infty$ to the set of CCEs.
\end{theorem}
\begin{proof}
	By definition of cumulative regret, and by taking its average, we have
	\[
	\limsup_{T\to\infty}\frac{1}{T}\max_{\hat\sigma_i\in \Sigma_i}\sum_{t=1}^T\left( u_i^t(\hat \sigma_i)- u_i^t(\sigma_i^t)\right)\leq 0,
	\] 
	which holds almost surely.
	Let $\sigma^t\coloneqq (\sigma_i^t,\sigma_{-i}^t)$. It follows that, for each normal-form plan $\hat \sigma_i\in \Sigma_i$ we have
	\begin{align*}
	&\frac{1}{T}\sum_{t=1}^T\left(u_i(\hat \sigma_i,\sigma_{-i}^t)-u_i(\sigma^t)\right)= \\
	&\hspace{3cm} \sum_{\sigma\in\Sigma} \bar x^T(\sigma)\left(u_i(\hat \sigma_i,\sigma_{-i})-u_i(\sigma)\right).
	\end{align*}
	
	Where $\bar x^T(\sigma)$ is the empirical frequency of $\sigma$ after $T$ iterations.
	On any subsequence where $\bar x^T$ converges, that is $\bar x^T\to x^\ast\in\X$, it holds almost surely, for each $\hat \sigma_i\in\Sigma_i$ that
	\begin{align*}
	&\sum_{\sigma\in\Sigma}\bar x^T(\sigma)\left(u_i(\hat \sigma_i,\sigma_{-i})-u_i(\sigma)\right)\to \\
	&\hspace{3cm}\sum_{\sigma\in\Sigma}x^\ast(\sigma)\left(u_i(\hat \sigma_i,\sigma_{-i})-u_i(\sigma)\right). 
	\end{align*}
	The result immediately holds for Definition~\ref{def:cce}.
\end{proof}

In the following, we follow the approach of~\citet{farina2018online} to show how to decompose $\tilde R_i^T$ into regret terms which are computed locally at player $i$'s infosets.
This allows us to avoid working with the (exponential-sized) normal form of an EFG even if $\tilde R_i^T$ is defined over player~$i$'s normal-form plans. 
%
%
$\tilde R_i^T$ can be minimized via the minimization of other suitably defined regrets computed locally at player $i$'s infosets.
In order to do this, we use the idea of \emph{laminar regret decomposition}~\cite{farina2018online}, but reasoning only on vertices of $\X_i$.

Given $\sigma_i \in \Sigma_i$, we denote by $\sigma_i(I)$ the action selected in $\sigma_i$ at infoset $I\in \mathcal{I}_i$. 
Moreover, $\sigma_{i \downarrow I}$ is the (sub)vector containing the actions selected in $\sigma_i$ at $I \in \mathcal{I}_i$ and all its descendant infosets.

First, we denote with $u_{i,I}^t: A(I)\to \mathbb{R}$ the \emph{immediate utility} observed by player $i$ at infoset $I \in \I_i$, during iteration $t$.
For every $a \in A(I)$, $u_{i,I}^t(a)$ is the utility experienced by player $i$ if the game ends after playing $a$ at $I$, without passing through another player $i$'s infoset. 

Then, the following is player $i$'s utility attainable at infoset $I\in\I_i$ when a normal-form plan $\hat\sigma_i \in \Sigma_{i}$ is selected:
\begin{equation}
\hat V_I^t(\hat \sigma_{i\downarrow I}) \coloneqq u^t_{i,I}(\hat\sigma_{i\downarrow I}(I)) + \sum_{I'\in\mathcal{C}_{I,\hat\sigma_{i\downarrow I}(I)}}\hat V_{I'}^t(\hat \sigma_{i\downarrow I'}),
\end{equation}
where $\mathcal{C}_{I,a} \subseteq \mathcal{I}_i$ is the set of possible next player $i$'s infosets, given that she played action $a \in A(I)$ at infoset $I \in \mathcal{I}_i$.
We introduce a parameterized utility function, which is used to define regrets locally at each infoset, and reads as follows:
\begin{equation}\label{eq:parametrized_util}
\hat u_{i,I}^t: a\in A(I)\mapsto u_{i,I}^t(a) + \sum_{I'\in\mathcal{C}_{I,a}} 
\hat V_{I'}^t(\sigma_{i\downarrow I'}^t).
\end{equation}

The utility function $\hat u_{i,I}^t$ preserves convexity of $u_{i}^t$.
Finally, we modify the notion of \emph{laminar regret}, as
\begin{equation}\label{eq:laminar}
\hat R_I^t\coloneqq\max_{ a\in A(I)}\sum_{t=1}^T\hat u^t_{i,I}( a) - \sum_{t=1}^T\hat u_{i,I}^t(\sigma_i^t(I)).
\end{equation}

Let $V_I^t\coloneqq \hat V_I^t(\sigma_{i,\downarrow I}^t)$.
Then, we introduce the cumulative regret at infoset $I\in\I_i$, defined as
\begin{equation}\label{eq:cumulative_regret_I}
R_{\downarrow I}^T\coloneqq 
\max_{\hat \sigma_{i \downarrow I}}\sum_{t=1}^T\hat V_{I}^t(\hat \sigma_{i\downarrow I})-\sum_{t=1}^T V^t_{I}.
\end{equation}

\begin{lemma}\label{lemma:cumulative}
	The cumulative regret at each infoset $I\in\I_i$ can be decomposed as
	\begin{equation*}
	R_{\downarrow I}^T=\!\max_{ a\in A(I)}\left(\sum_{t=1}^T\hat u_{i,I}^t( a)\right.\left.+\!\!\!\sum_{I'\in\mathcal{C}_{I, a}} \!\!\!\! R_{\downarrow I'}^T\!\right)-\sum_{t=1}^T\hat u_{i,I}^t(\sigma_{i}^t(I)).
	\end{equation*}
\end{lemma}
\begin{proof}
	By definition of cumulative regret at $I\in\I_i$ we have that:
	\begin{align*}
		&R_{\downarrow I}^T\coloneqq  
		\max_{\hat \sigma_{i \downarrow I}}\sum_{t=1}^T\hat V_{I}^t(\hat \sigma_{i \downarrow I})-\sum_{t=1}^TV^t_{I}=\\
		& =\max_{\hat \sigma_{i\downarrow I}}\sum_{t=1}^T\left(u^t_{i,I}(\hat\sigma_{i \downarrow I}(I)) + \sum_{I'\in\mathcal{C}_{I,\hat\sigma_{i \downarrow I}^t(I)}} \hat V_{I'}^t(\hat \sigma_{i \downarrow I'}) \right) -\\
		&\hspace{6cm} \sum_{t=1}^T V^t_{I}=\\
		& = \max_{ a\in A(I)}\left( \sum_{t=1}^T u^t_{i,I}( a)+ \sum_{I'\in\mathcal{C}_{I,a}}\max_{\hat\sigma_{i \downarrow I'}}\sum_{t=1}^T\hat V_{I'}^t(\hat\sigma_{i \downarrow I'})\right)-\\
		&\hspace{6cm} \sum_{t=1}^T V^t_{I}. \\
	\end{align*}
	Then, by employing Equation~\eqref{eq:cumulative_regret_I}, we get
	\begin{align*}
		&R_{\downarrow I}^T=\max_{ a\in A(I)}\left( \sum_{t=1}^T u^t_{i,I}( a) +\right.\\
		&\hspace{2cm}\left.\sum_{I'\in\mathcal{C}_{I,a}}\left(R_{\downarrow I'}^T+\sum_{t=1}^T V_{I'}^t\right)\right)- \sum_{t=1}^T V^t_{I}.
	\end{align*}
	Finally, we obtain the result by rewriting terms according to Equation~\eqref{eq:parametrized_util}.
\end{proof}

The following theorem shows that, in order to minimize $\tilde R_i^T$, it is enough to minimize the laminar regret locally at each $I\in\I_i$ as defined in Equation~\eqref{eq:laminar}. 

\begin{lemma}\label{lemma:sum}
	The cumulative regret $\tilde R_i^T$ satisfies the following:
	\[
	\tilde R_i^T \leq \max_{\hat \sigma_i \in \Sigma_i}\sum_{I\in\I_i}\rho^{\hat \sigma_i}_I\hat R_I^T.
	\]
\end{lemma}


\begin{proof}
	Consider a generic infoset $I\in\I_i$. By exploiting Lemma~\ref{lemma:cumulative} and Definition~\ref{eq:laminar}, we can write:
	\begin{align*}
	&R_{\downarrow I}^T=\max_{ a\in A(I)}\left(\sum_{t=1}^T\hat u_{i,I}^t( a) +\sum_{I'\in\mathcal{C}_{I, a}} R_{\downarrow I'}^T\right)- \\
	&\hspace{5cm} \sum_{t=1}^T\hat u_{i,I}^t(\sigma_{i}^t(I))\leq
	&\\
	&\leq \!\! \max_{ a\in A(I)} \! \sum_{t=1}^T  \hat u_{i,I}^t( a) \! + \!\!\! \max_{ a\in A(I)} \!
	\!\sum_{I'\in\mathcal{C}_{I, a}} \!\!\!\!\! R_{\downarrow I'}^T \! -\!\sum_{t=1}^T \hat u_{i,I}^t(\sigma_{i}^t(I)) \! = \\
	& =\hat R_{I}^T+\max_{ a\in A(I)}
	\sum_{I'\in\mathcal{C}_{I, a}} R_{\downarrow I'}^T.
	\end{align*}
	By starting from the root of the game and applying the above equation inductively, we obtain our result. 
\end{proof}

The last result provides an immediate proof of the following.

\cfrs*
\begin{proof}
	CFR-S minimizes each laminar regret $\hat R_I^T$, as defined in Equation~\eqref{eq:laminar}, through standard RM, which guarantees that $\limsup_{T\to\infty} \frac{1}{T} \hat R_I^T\leq 0$ almost surely.
	Therefore, $\limsup_{T\to\infty} \frac{1}{T} \tilde R_i^T\leq 0$ almost surely (Lemma~\ref{lemma:sum}), which implies that the empirical frequency of play converges almost surely to a CCE for $T\to \infty$ (Theorem~\ref{thm:convergence-to-cce}).
\end{proof}

Finally, we observe that, at each iteration $t$ and infoset $I \in \I_i$, $\sigma_i^t(I)$ is selected according to the strategy $\pi_{i,I}^t$ recommended by the regret minimizer at infoset $I$.
Thus, $\sigma_i^t$ is drawn with probability $\prod_{I\in\I_i}\pi_{i,I,\sigma_{i}^t(I)}^t$, which is equal to $x_i^t(\sigma_i^t)$, where $x_i^t\in\X_i$ is the normal-form strategy realization equivalent to the behavioral strategy $\pi_i^t$.

%% file: content/appendix_cfr_jr.tex
\section{CFR-Jr}\label{appendix:cfr-jr}

\subsection{Algorithm Pseudocode}

For the sake of completeness and clarity, in Algorithm~\ref{alg:cfr_jr} we provide the pseudocode of the CFR-Jr algorithm, which uses a vanilla implementation of the CFR algorithm as a subroutine.

\begin{algorithm}[H]\label{algo:cfr_jr}
	\centering
	\scriptsize
	\caption{\texttt{CFR-Jr}}
	\begin{algorithmic}[1]
		\Function{\textsc{CFR-Jr}}{$\Gamma$}
			\State Initialize the joint strategy $\bar{x}$ to all zeros
			\State $t\gets 0$
			\While{$t < T$}
				\ForAll{$i \in \mathcal{P}$}
					\State $\pi_i^t \gets \textsc{CFR}(\Gamma, i)$\label{line:cfr}
					\State $x_i^t \gets \textsc{NF-Strategy-Reconstruction}(\pi_i^t)$\label{line:recon}
				\EndFor
				\State $\bar{x} \gets \bar{x} +  \bigotimes_{i \in \mathcal{P}} x_i^t$\Comment{$\bigotimes_{i \in \mathcal{P}} x_i^t$ is joint distribution $x^t$ defined as the product of the players' normal-form strategies}\label{line:sum_x_t}
				\State$t\gets t+1$
			\EndWhile
			\State \Return $\bar{x}^T = \frac{\bar{x}}{T} $
		\EndFunction
	\end{algorithmic}
	\label{alg:cfr_jr}
\end{algorithm}

CFR-Jr maintains a variable $\bar x$ which stores the sum of the joint probability distributions $x^t$ (notice that it may be compactly represented in polynomial space using a dictionary, as for $x_i$ in Algorithm~\ref{alg:algo_recon}). 
CFR-Jr executes, at each $t$, an iteration of the CFR algorithm (Line~\ref{line:cfr}). 
In particular, the $\textsc{CFR}$ subroutine executes a step of vanilla CFR, including the update of regrets and behavioral strategies. 
In addition, at each iteration $t$, CFR-Jr constructs normal-form strategies $x_i^t$ (one per player $i \in \mathcal{P}$) which are realization equivalent to the behavioral strategies $\pi_i^t$ obtained with CFR (Line~\ref{line:recon}).
Then the product $x^t$ of the players' normal-strategies is computed and added to $\bar x$ (Line~\ref{line:sum_x_t}).
Notice that $\bar{x}$ is not used by the $\textsc{CFR}$ subroutine to update the players' strategies and regrets. Finally, CFR-Jr returns the $\bar{x}$ divided by $T$, which represents the average $\bar x^T$.

\subsection{Omitted Proofs}

In order to give the full proof of Theorem~\ref{thm:corr_algo_rec}, we first need to prove two technical lemmas concerning the existence of a normal-form plan $\bar{\sigma}_i$ such that $\bar{\omega} = \min_{z \in Z(\bar{\sigma}_i)} \omega_z > 0$ whenever the vector $\omega$ has at least a strictly positive component.

For ease of presentation, we introduce some additional notation.
Extending the definition of $Z(\sigma_{i})$, we let $Z(I, a)$ be the set of terminal nodes potentially reachable from infoset $I \in \mathcal{I}_i$ when playing action $a \in A(I)$, while $Z(\sigma_i, I, a)$ is the set of terminal nodes potentially reachable from $I$ after playing action $a$ and, then, following the actions prescribed by the normal-form plan $\sigma_i \in \Sigma_i$.
Analogously, we define $Z(I)$ and $Z(\sigma_i, I)$.

Observe that, in Line~\ref{line:recon_maxmin} of Algorithm~\ref{alg:algo_recon}, the normal-form plan $\bar{\sigma}_i  \in \argmax_{\sigma_i \in \Sigma_i} \min_{z \in Z(\sigma_i)} \omega_z$ can be recursively built, by traversing the game tree, as the set of actions $\{\bar{a} \in A(I) \mid I \in \mathcal{I}_i \wedge \bar{a} \in \argmax_{a \in A(I)} \min_{z \in Z(\bar{\sigma}_i, I, a)} \omega_z \}$. 
This recursion is heavily exploited in the following proofs.

\begin{lemma}\label{induction_lemma}
	Given $\bar{\sigma}_i\in\Sigma_i$ constituted by actions $\{\bar{a} \in A(I) \mid I \in \mathcal{I}_i \wedge \bar{a} \in \argmax_{a \in A(I)} \min_{z \in Z(\bar{\sigma}_i, I, a)} \omega_z \}$, for every infoset $I \in \mathcal{I}_i$ we have that:
	$$
	\max_{a \in A(I)} \min_{z \in Z(\bar{\sigma}_i, I, a)} \omega_z = 0 \,\, \Leftrightarrow \,\, \omega_z = 0 \,\,\,\, \forall z \in Z(I).
	$$
\end{lemma}

\begin{proof}
	The proof is by induction on the depth of the game tree.
	Let $\mathcal{C}_{I, a}$ be the set of player $i$'s infosets immediately reachable by playing action $a \in A(I)$ at infoset $I \in \I_i$.
	
	As for the base case of the induction, let us consider an infoset $I \in \mathcal{I}_i$ such that $\mathcal{C}_{I, a} = \varnothing $ for all $ a \in A(I)$ (\emph{i.e.}, one such that no other infosets of player $i$ may be reached after). 
	Observe that, due to $\omega$ being initialized using only player $i$'s behavioral strategy $\pi_i$, we have that $\omega_z = \omega_{z'} = \rho^{\pi_i}_I \pi_{i,I,a}$ for all $ a \in A(I)$ and $z, z' \in Z(I, a)$, which in turn implies that $\max_{a \in A(I)} \min_{z' \in Z(\bar{\sigma}_i, I, a)} \omega_{z'} = \max_{a \in A(I)} \rho^{\pi_i}_I \pi_{i,I,a}= \max_{z' \in Z(I)} \omega_{z'}$. 
	The max of a non-negative function over a set is zero if and only if that function is zero for all the elements of the set, thus $\max_{a \in A(I)} \min_{z' \in Z(\bar{\sigma}_i, I, a)} \omega_{z'} = \max_{z' \in Z(I)} \omega_{z'} = 0$ if and only if $\omega_z = 0 $ for all $ z \in Z(I)$.

	For the inductive step, let us consider a generic infoset $I \in \mathcal{I}_i$. Observe that, if $\bar{a} \in \argmax_{a \in A(I)} \min_{z' \in Z(\bar{\sigma}_i, I, a)} \omega_{z'}$, we have $Z(\bar{\sigma}_i, I, \bar{a}) = Z(\bar{\sigma}_i, I)$ (since $\bar{a} \in \bar{\sigma}_i$). 
	Moreover, reasoning as above, we can conclude that 
	$$\max_{a \in A(I)} \min_{z' \in Z(\bar{\sigma}_i, I, a)} \omega_{z'} = 0$$
	if and only if $\min_{z' \in Z(\bar{\sigma}_i, I, a)} \omega_{z'} = 0 $ for all $a \in A(I)$.
	Now let us take any action $a' \in A(I)$ and any infoset $I' \in \mathcal{C}_{I, a'}$ (it has to exist at least one, otherwise we would fall back in the base case). Applying the observation above, we have that
	$$\max_{a \in A(I')} \, \min_{z' \in Z(\bar{\sigma}_i, I', a)} \omega_{z'} = \!\!\!\! \min_{z' \in Z(\bar{\sigma}_i, I', \bar{a})} \! \omega_{z'} = \! \min_{z' \in Z(\bar{\sigma}_i, I')} \omega_{z'}.$$
	But, being $I'$ a descendant of $I$ we have that $Z(I') \subseteq Z(I, a')$ and, in particular, also $Z(\bar{\sigma}_i, I') \subseteq Z(\bar{\sigma}_i, I, a')$.
	Thus, $\min_{z' \in Z(\bar{\sigma}_i, I, a')} \omega_i(z') = 0$ implies that
	$$\min_{z' \in Z(\bar{\sigma}_i, I')} \omega_i(z') = 0.$$ 
	By the induction hypothesis, we have that for $I'$ (which follows $I$ in the game tree) it holds
	$$\max_{a \in A(I')} \min_{z' \in Z(\sigma^*, I', a)} \omega_i(z') = 0$$
	if and only if $\omega_i(z) = 0 $ for every $ z \in Z(I')$. We can now build our final chain of double implications:
	$$\max_{a \in A(I)} \min_{z' \in Z(\bar{\sigma}_i, I, a)} \!\!\!\! \omega_{z'} = 0 \, \Leftrightarrow \!\! \min_{z' \in Z(\bar{\sigma}_i, I, a)} \!\!\!\!\! \omega_{z'} = 0 \,\,\,\, \forall a \in A(I),$$
	where, by generality of $a'$ and $I'$, the right term holds if and only if 
	$$\max_{a \in A(I')} \min_{z' \in Z(\bar{\sigma}_i, I', a)} \omega_{z'} = 0 \,\,\,\, \forall a' \in A(I), I' \in \mathcal{C}_{I, a'},$$
	which in turn holds, by inductive hypothesis, if and only if $$\omega_z = 0 \,\,\,\, \forall z \in Z(I') \,\,\,\, \forall a' \in A(I), I' \in \mathcal{C}_{I, a'}.$$
	This holds, given that $\bigcup_{a' \in A(I), I' \in \mathcal{C}_{I, a'}} Z(I') = Z(I)$, if and only if $\omega_z = 0 $ for all $ z \in Z(I)$, which concludes the proof.
\end{proof}

\begin{lemma}
	If $ \omega > 0$, then a normal-form plan $\bar{\sigma}_i \in \argmax_{\sigma_i \in \Sigma_i} \min_{z \in Z(\sigma_i)} \omega_z$ is such that $\min_{z \in Z(\bar{\sigma}_i)} \omega_z > 0$.
\end{lemma}

\begin{proof}
	Let $\varnothing$ be the root infoset of player $i$.
	Observe that, if $\bar{a} \in \argmax_{a \in A(\varnothing)} \min_{z' \in Z(\bar{\sigma}_i, \varnothing, a)} \omega_{z'}$, we have $Z(\bar{\sigma}_i) = Z(\bar{\sigma}_i, \varnothing, \bar{a})$ (since $\bar{a} \in \bar{\sigma}_i$).
	Applying Lemma~\ref{induction_lemma} to infoset $\varnothing$, we have that $\max_{a \in A(\varnothing)} \min_{z' \in Z(\bar{\sigma}_i, \varnothing, a)} \omega_{z'} = \min_{z' \in Z(\bar{\sigma}_i)} \omega_{z'} = 0$ if and only if $\omega_z = 0 $ for every $z \in Z(\varnothing) = Z$. 
	Since $\omega_z \geq 0 $ for all $ z \in Z$, we have $\min_{z' \in Z(\bar{\sigma}_i)} \omega_{z'} > 0$ if and only if $\omega > 0$. Being this last condition always verified within the main loop of Algorithm~\ref{alg:algo_recon}, we have that a normal-form plan $\bar{\sigma}_i \in \argmax_{\sigma_i \in \Sigma_i} \min_{z \in Z(\sigma_i)} \omega_z$ is such that $\min_{z \in Z(\bar{\sigma}_i)} \omega_z > 0$.
\end{proof}

\thmAlgoRecon*

\begin{proof}
	Due to $\bar{\omega} \, \rho_z^{\bar{\sigma}_i}$ being equal to $\bar{\omega} = \min_{z \in Z(\bar{\sigma}_i)} \omega_z$ for all $z \in Z(\bar{\sigma}_i)$ and equal to zero for all $z \not\in Z(\bar{\sigma}_i)$, we have that after each iteration all components of $\omega$ are non-negative. 
	Moreover, at least for $\bar{z} \in \argmin_{z \in Z(\bar{\sigma}_i)} \omega_z$, we have $\omega_{\bar{z}} = 0$.
	Then, at each iteration, $\omega$ becomes zero for at least one terminal node, and, given that it cannot go below zero, we have that the vector $\omega$ is zeroed in at most $|Z|$ iterations.
	Each iteration runs in $O(\max\{|\mathcal{I}_i|, |Z|\})$, as $\bar{\sigma}_i$ can be recursively computed by iterating on the infosets in a bottom-up fashion, while each $\omega$ update needs to consider each terminal at most once. 
	Given that for each non-degenerate game tree (\emph{i.e.}, $A(I) > 1 $ for all $ I \in \mathcal{I}_i$) we have $|\mathcal{I}_i| \leq |Z|$, the overall complexity of the algorithm is $O(|Z|^2)$.
	
	As a consequence of $\omega$ becoming zero for some $z \in Z(\bar{\sigma}_i)$ at each iteration, and given that $\bar{\sigma}_i$ is built in such a way that $\omega_z > 0 $ for all $z \in Z(\bar{\sigma}_i)$, no normal-form plan is ever selected more than once as $\bar{\sigma}_i$. 
	Then, the support of $x_i$ has size equal to the number of normal-form plans $\bar{\sigma}_i$ selected at each iteration of Algorithm~\ref{alg:algo_recon}, which is at most $|Z|$.
	
	Let $\bar{k}$ be the number of iterations, and $\bar \sigma_i^k$ be the normal-form plan reconstructed at iteration $k$.
	By recursively expanding the equation $\omega_z \leftarrow \omega_z - \bar{\omega} \, \rho_z^{\bar{\sigma}_i}$ we can obtain the following (for clarity, we add apices indicating the iteration):
	\begin{align*}
		\omega_z^{\bar{k}} & = \omega_z^{\bar{k}-1} - \rho_z^{\bar \sigma_i^{\bar{k}}} \!\! \min_{z' \in Z(\bar \sigma_i^{\bar{k}-1})} \!\! \omega_{z'}^{\bar{k}-1} =
		\\
		&= \omega_z^{\bar{k}-2} - \rho_z^{\bar \sigma_i^{\bar{k}-1}} \!\!\!\!\!\! \min_{z' \in Z(\bar \sigma_i^{\bar{k}-2})} \!\! \omega_{z'}^{\bar{k}-2} - \rho_z^{\bar \sigma_i^{\bar{k}}} \!\!\!\! \min_{z' \in Z(\bar \sigma_i^{\bar{k}-1})} \!\! \omega_{z'}^{\bar{k}-1} =
		\\
		& = \ldots = \omega_z^0 - \sum_{k = 1,..,\bar{k}} \rho_z^{\bar \sigma_i^{k}} \!\! \min_{z' \in Z(\bar \sigma_i^{k-1})} \!\! \omega_{z'}^{k-1} = 0,
	\end{align*}
	which gives
	$$
	\omega_z = \sum_{k = 1,..,\bar{k}} \rho_z^{\bar\sigma_i^{k}} \min_{z' \in Z(\bar \sigma_i^{k-1})} \omega_{z'}^{k-1}.
	$$
	Finally, we show that $x_i$ and $\pi_i$ are realization equivalent by checking that they force the same probability distribution over $Z$.
	We have:
	\begin{align*}
		&\sum_{\sigma_i \in \Sigma_i} \! \rho_z^{\sigma_i} x_i(\sigma_i) = \hspace{-.5cm}\sum_{\sigma_i \in \{\bar \sigma_i^k \mid k = 1,..,\bar{k}\}} \hspace{-.75cm}\rho_z^{\sigma_i} x_i(\sigma_i) = \\
		&\sum_{k = 1,..,\bar{k}} \!\! \rho_z^{\bar \sigma_i^k} x_i(\bar \sigma_i^k) = \sum_{k = 1,..,\bar{k}} \!\! \rho_z^{\bar\sigma_i^k} \min_{z' \in Z(\bar \sigma_i^{k-1})} \omega_{z'}^{k-1} = \rho_z^{\pi_i}.
	\end{align*}
	This concludes the proof.
\end{proof}

\thmCFRJr*

\begin{proof}
	First, let us recall that $x^t \in \X$ is defined in such a way that $x^t(\sigma) = \prod_{i \in \mathcal{P}} x_i^t(\sigma_i)$ for every joint normal-form plan $\sigma \in \Sigma$, with $\sigma = (\sigma_i)_{i \in \mathcal{P}}$.
	By assumption, $\frac{1}{T}R_i^T \leq \varepsilon$ implies the following:
	\begin{align*}
		&\max_{\hat \sigma_i \in \Sigma_i} \left( \sum_{t=1}^T \sum_{\sigma_{-i} \in \Sigma_{-i}} u_i(\hat \sigma_i, \sigma_{-i}) \prod_{j \neq i \in \mathcal{P}} x_j^t(\sigma_j) - \right.\\
		&\hspace{1.2cm} \left. \sum_{t=1}^T \sum_{\substack{\sigma_i \in \Sigma_i \\ \sigma_{-i}\in \Sigma_{-i}}} u_i(\sigma_i,\sigma_{-i}) \prod_{j  \in \mathcal{P}} x_j^t(\sigma_j)  \right)\leq \varepsilon T.
	\end{align*}
	Moreover, since the condition holds for every $i \in \mathcal{P}$, by re-writing the max operator we get
	\begin{align*}
		&\sum_{t=1}^T \sum_{\sigma_{-i} \in \Sigma_{-i}} u_i(\hat \sigma_i, \sigma_{-i}) \prod_{j \neq i \in \mathcal{P}} x_j^t(\sigma_j) -\\
		&\sum_{t=1}^T \sum_{\substack{\sigma_i \in \Sigma_i \\ \sigma_{-i}\in \Sigma_{-i}}} \!\!\! u_i(\sigma_i,\sigma_{-i}) \prod_{j  \in \mathcal{P}} x_j^t(\sigma_j)  \leq \varepsilon T \,\,\,\, \forall i \in \mathcal{P}, \hat \sigma_i \in \Sigma_i.
	\end{align*}
	Since $\sum_{\sigma_i \in \Sigma_i} x_i^t(\sigma_i) = 1$, it follows that
	$$\sum_{\sigma_{-i} \in \Sigma_{-i}} u_i(\hat \sigma_i, \sigma_{-i}) \prod_{j \neq i \in \mathcal{P}} x_j^t(\sigma_j)$$ i
	s equal to 
	$$\sum_{\sigma_i \in \Sigma_i} \sum_{\sigma_{-i} \in \Sigma_{-i}} u_i(\hat \sigma_i, \sigma_{-i}) \prod_{j \in \mathcal{P}} x_j^t(\sigma_j)$$. Thus,
	\begin{align*}
		&\sum_{t=1}^{T} \sum_{\substack{\sigma_i \in \Sigma_i \\ \sigma_{-i}\in \Sigma_{-i}}} \prod_{j \in \mathcal{P}} x_j^t(\sigma_{j}) \left( u_i(\hat \sigma_i, \sigma_{-i}) - \right. \\
		&\hspace{2cm} \left. u_i(\sigma_i, \sigma_{-i})  \right) \leq \varepsilon T \,\,\,\, \forall i \in \mathcal{P}, \hat \sigma_i \in \Sigma_i.
	\end{align*}
	Using the definition of $\bar x^T$, we obtain
	\begin{align*}
	&\sum_{\substack{\sigma_i \in \Sigma_i \\ \sigma_{-i}\in \Sigma_{-i}}} \bar x^T(\sigma_i, \sigma_{-i}) \left( u_i(\hat \sigma_i, \sigma_{-i}) - \right. \\
	&\hspace{2cm} \left. u_i(\sigma_i, \sigma_{-i})  \right) \leq \varepsilon \,\,\,\, \forall i \in \mathcal{P}, \hat \sigma_i \in \Sigma_i,
	\end{align*}
	which proves that $\bar x^T$ is an $\varepsilon$-CCE.
\end{proof}

%% file: content/appendix_exp.tex
\section{Additional Details on the Experimental Evaluation}\label{sec:appendix_exp}

In this section we provide further details on the experimental evaluation.

\subsection{Experimental Setup}
The multi-player games instances that we employ are structured as follows.

\textbf{Kuhn Poker}.
In Kuhn3-$r$ (K3-$r$), each player initially pays one chip to the pot, and is dealt a single private card.
Then, players act in turns. 
The first player may check or bet, that is paying an additional chip to the pot.
The second player can either fold/call the bet, or check/bet after an initial check of the first player.
At this point, if no bets have been placed, the third player decides between checking or betting. 
Otherwise, she can either fold or call.
If the third player bets, then the others have to choose between folding or calling.
At the showdown, the player with the highest card who has not folded wins all the chips in the pot.

\textbf{Leduc Hold’em Poker}. We employ larger three-player variants than the two-player version usually employed, see, e.g., \cite{southey2005bayes}. 
In our enlarged variants, Leduc3-$r$ (L3-$r$) contains three suits and $r \geq 3$ card ranks (i.e., it contains triples of cards $A,2,\ldots,r$ for a total of $3\,r$ cards).
Each player initially pays one chip to the pot, and is dealt a single private card. 
After a first round of betting (with betting parameter $k_1$), a community card is dealt face up. 
Then, a second round of betting is played (with betting parameter $k_2$). 
Finally, a showdown occurs and players that did not fold reveal their private cards. 
If a player pairs her card with the community card, she wins the pot. 
Otherwise, the player with the highest private card wins. 
If more players pair their card with the community card, or if no pair is made but more players have the same private card, there is a draw and the pot is equally split between the winning players.
Betting rounds follow the same rules of Kuhn Poker. 
We set $k_1=2$ and $k_2=4$. These are the numbers of chips that a player has to pay to bet/call in the first and second round of betting, respectively.

\textbf{Goofspiel}. In addition to Poker games, we consider the game of Goofspiel. 
In this game, cards rank A (low), 2, $\ldots$, 10, J, Q, K (high). 
When scoring points, the Ace is worth 1 point, cards 2-10 their face value, Jack 11, Queen 12, and King 13.
Goofspiel$p$-$r$ ($p$ is the number of players) employs $p+1$ suits, each containing cards A$,\ldots$,$r$.
One suit is singled out as the prizes. The prizes are shuffled and placed between the players, with the top card turned face up. 
Each of the remaining suits becomes the hand of one of the players.
The game proceeds in rounds. Each player selects a card from her hand, keeping her choice secret from the opponent. Once all players have selected a card, they are simultaneously revealed, and the player with the highest bid wins the prize card.
We employ the following tie breaking rules to obtain different kinds of instances.
Some of them (\emph{e.g.}, \emph{Accumulate}) are \emph{almost} constant-sum games (\emph{i.e.}, constant sum for all but few outcomes), while others (\emph{e.g.}, \emph{Discard always}) present larger differences in the sum of payoffs attainable at different terminal nodes:
\begin{itemize}[nolistsep,itemsep=0mm]
	\item \emph{Accumulate} (A): the prize card goes to the player that selected the highest unique card.
	If all players selected the same card, the prize card is taken aside and the game continues unveiling the next one: the winner (if any) of the new round will take both prize cards.
	The process is repeated until the tie is broken or the game ends, in which case all prize cards that have been taken aside are discarded.
	
	\item \emph{Discard-if-all} (DA): the prize card goes to the player that selected the highest unique card; if all players selected the same card, the prize card is discarded.
	
	\item \emph{Discard-if-high} (DH): if the tie is on the highest-valued card, then the prize card is discarded; otherwise, the prize card goes to the player that selected the highest unique card.
	
	\item \emph{Discard always} (AL): the prize card is discarded and the game goes on with the next round.
\end{itemize}

The game ends when the players terminate their cards. Players calculate their final
utility by summing up the value of the prize cards they won.

\subsection{Full Experimental Results}

%
CFR, CFR-S and CFR-Jr algorithms have all been implemented in the Python 3 language, while CG employs the GUROBI 8.0 MILP solver to solve the pricing problems. All the experiments are run with a 24 hours time-limit on a UNIX machine with a total of 32 cores working at 2.3 GHz, equipped with 128 GB of RAM.

\begin{table*}
	\centering
	\includegraphics[width=0.98\textwidth]{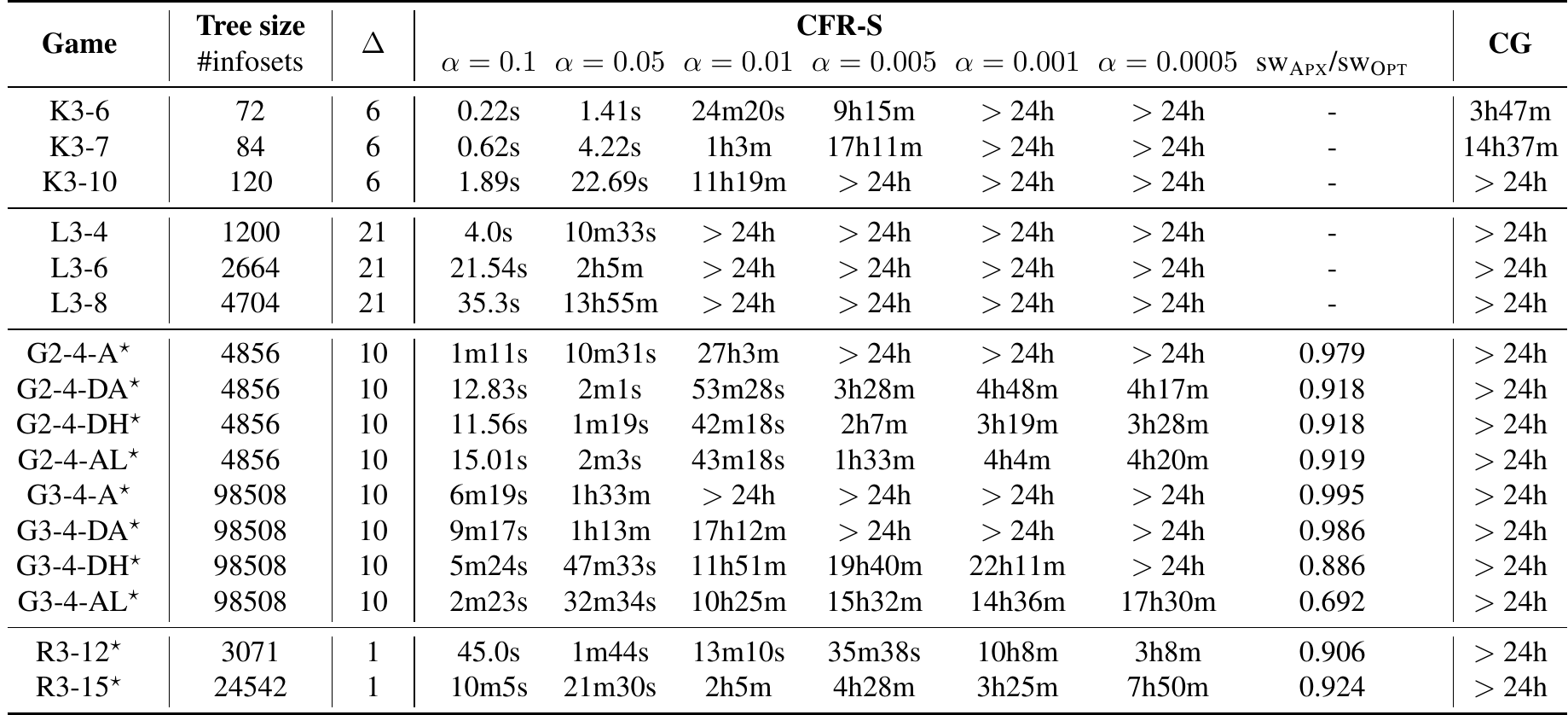}
	\caption{Running times and social welfare obtained by the CFR-S algorithm (for various levels of accuracy), and the CG algorithm. General-sum instances are marked with $^\star$.}
	\label{tab:full_convergence_cfr_s}
\end{table*}

\begin{table*}
	\centering
	\includegraphics[width=0.98\textwidth]{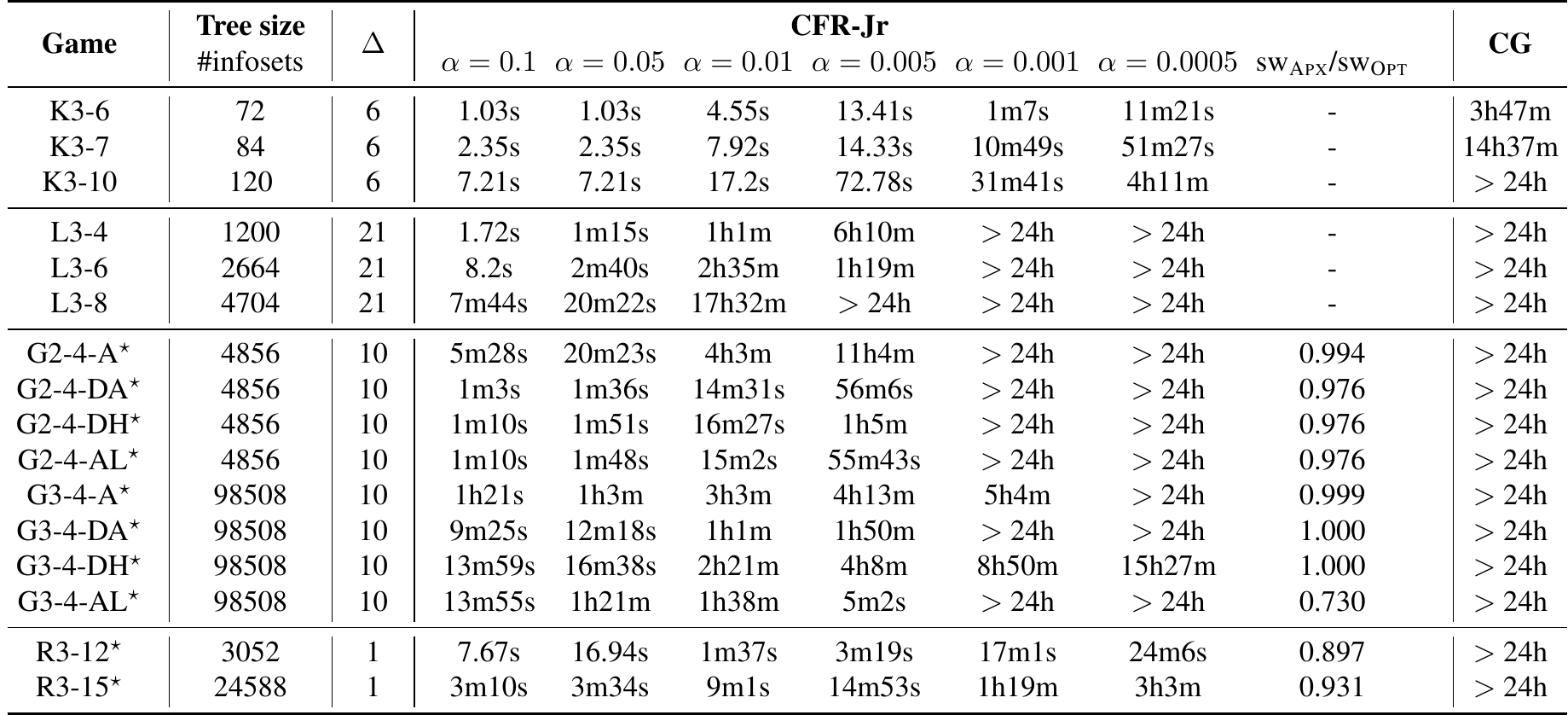}
	\caption{Running times and social welfare obtained by the CFR-Jr algorithm (for various levels of accuracy), and the CG algorithm. General-sum instances are marked with $^\star$.}
	\label{tab:full_convergence_cfr_jr}
\end{table*}

We tested CFR-Jr also on some extensive-form variants of the Shapley game, a normal-form general-sum 3x3 game that has been shown to induce cyclic non-convergent behaviors in iterative algorithms such as Fictitious Play~\cite{jafari2001no}. 
The results in Figures~\ref{fig:shapley_game}--\ref{fig:shapley_game_sw} clearly show that also CFR can get stuck in non-convergent cycles, confirming what we observed in Figure~\ref{fig:goof}.
This is well known in theory (there is no guarantee of convergence for CFR in general-sum games, even if two-player and no chance) but, to the best of out knowledge, was never observed in practice. 
Note that also CFR-S has some difficulties in reaching low values of $\varepsilon$, while CFR-Jr reaches a good approximation of an equilibrium point in few iterations. 
The game we used for this experiment is structured as follows:

\begin{itemize}
	\item There is a set of three cards, numbered 0 to 2.
	\item Player 1 plays a card out of the set, then Player 2 plays a hidden card out of the set, and finally Player 1 plays again a card out of the three.
	\item If the sum of the cards that have been played is $0 \bmod 3$, then the utility is $(0, 0)$, if the sum is $1 \bmod 3$, it is $(1, 0)$, and if it is $2 \bmod 3$, then the players receive a utility of $(0, 1)$.
	\item If the two cards played are equal in value, than the utility gained by each of the played is doubled (this step is fundamental to ensure that the uniform probability over the joint set of actions is not a CCE).
\end{itemize}


\begin{figure*}[t]
	\hspace{0.3cm}
	\begin{minipage}{.43\textwidth}
		\centering
		\includegraphics[width=\textwidth]{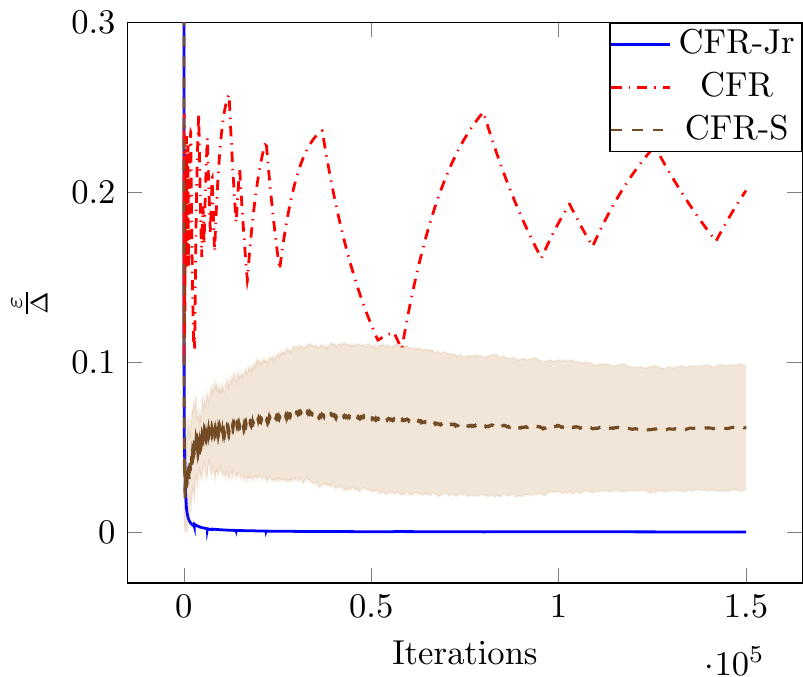}
		\caption{Convergence in number of iterations for the Shapley game}
		\label{fig:shapley_game}
	\end{minipage}
	\hspace{0.8cm}
	\begin{minipage}{.43\textwidth}
		\centering
		\includegraphics[width=\textwidth]{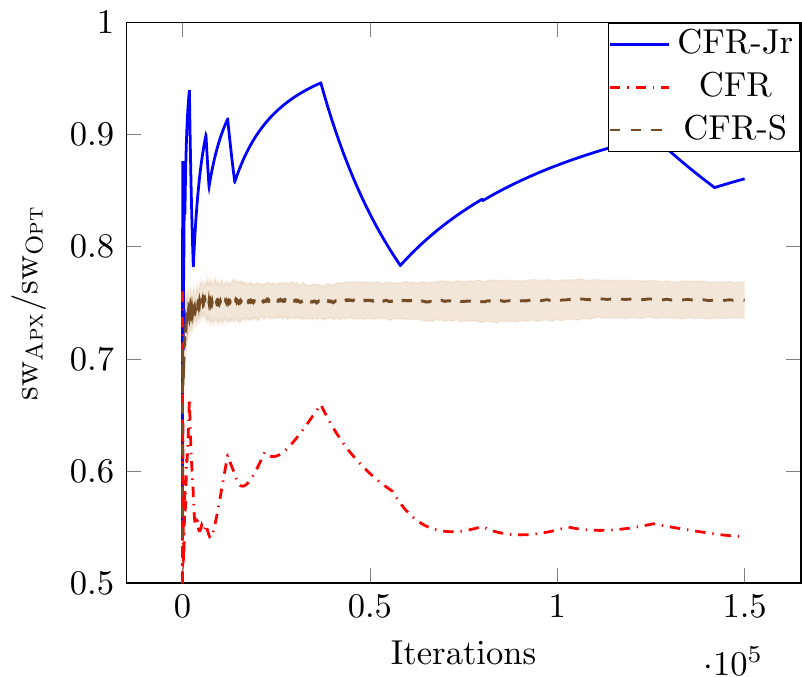}
		\caption{Social welfare attained with respect to the optimal one for the Shapley game}
		\label{fig:shapley_game_sw}
	\end{minipage}
	\hspace{0.3cm}
\end{figure*}

\subsection{CFR-Jr With Different Joint Distribution Reconstruction Rates}

From a theoretical standpoint, CFR-Jr requires a joint distribution reconstruction step to be carried out at every iteration $t$, to ensure that the resulting normal-form joint strategy approaches the set of CCEs (see Theorem~\ref{thm:eps_cce_avg_prod}). 
%
%
We investigate whether it is possible to trade some accuracy to reduce the computational time of the algorithm, by performing the joint distribution reconstruction at a subset of the iterations. 
This could also allow the algorithm to store smaller normal-form strategies, by skipping the reconstruction during the first iterations.
Indeed, during the first iterations, CFR (and, therefore, CFR-Jr) produces behavioral strategies that tend to be fairly uniformly distributed over all the possible actions, leading to $\omega$s assigning some probability to all terminals.
This implies the reconstructed normal-form strategies, in the first iterations, have considerably large supports.

\begin{algorithm}[H]
	\centering
	\scriptsize
	\caption{\texttt{CFR-Jr-$k$}}
	\begin{algorithmic}[1]
		\Function{\textsc{CFR-Jr}}{$\Gamma$}
			\State Initialize the joint strategy $\bar{x}$ to all zeros
			\State $t\gets 0$
			\While{$t < T$}
				\ForAll{$i \in \mathcal{P}$}
					\State $\pi_i^t \gets \textsc{CFR}(\Gamma, i)$
					\If{$t \bmod k = 0$}					
						\State $x_i^t \gets \textsc{NF-Strategy-Reconstruction}(\pi_i^t)$
					\EndIf
				\EndFor
					\If{$t \bmod k = 0$}					
				\State $\bar{x} \gets \bar{x} +  \bigotimes_{i \in \mathcal{P}} x_i^t$\Comment{$\bigotimes_{i \in \mathcal{P}} x_i^t$ is joint distribution $x^t$ defined as the product of the players' normal-form strategies}
					\EndIf
				\State $t\gets t+1$
			\EndWhile
			\Return $\bar{x}^T = \frac{\bar{x}}{\left\lfloor\frac{T}{k}\right\rfloor} $
		\EndFunction
	\end{algorithmic}
	\label{alg:cfr_jr_k}
\end{algorithm}

These considerations suggest that a slight modification of the CFR-Jr algorithm, that we call CFR-Jr-$k$ (see Algorithm~\ref{alg:cfr_jr_k}), may perform better in some settings. 
The idea behind CFR-Jr-$k$ is that the reconstruction procedure is carried out only every $k$ iterations. 
We have evaluated CFR-Jr-$k$ for different values of $k$. 
In all the tests we performed, the CFR-Jr algorithm always showed good convergence. 
In Figures~\ref{fig:reconstruction_rate1}--\ref{fig:reconstruction_rate3}, we report the experimental results related to instances of Kuhn3-6. 
The plots show both the convergence speed in terms of number of iterations and in terms of run time, as well as the size of the support of the average joint strategy that was stored by the algorithm (which is always monotonically increasing by construction). 
In Figures~\ref{fig:reconstruction_rate1_k_3_10}--\ref{fig:reconstruction_rate3_k_3_10}, we report the experimental results related to instances of Kuhn3-10.

In regard to the performances, larger reconstruction rates let the algorithm complete the same amount of iterations in a shorter time.
On the other hand, smaller reconstruction rates can lead earlier to a good joint strategy, and hence to reach lower values of $\varepsilon$. 
There is a trade-off between iteration speed and reconstruction accuracy, which can be exploited to tackle different problems with the most suited level of precision.

For what regards the size of the support of the joint average strategy, we can clearly see that lower reconstruction rates, running more times the reconstruction algorithm in the same amount of time, and being more susceptible to high-frequency variations in the behavioral strategies built by CFR, require up to ten time more space to store their joint strategies.

\begin{figure*}[t]
	\begin{minipage}{.3\textwidth}
		\centering
		\includegraphics[width=\textwidth]{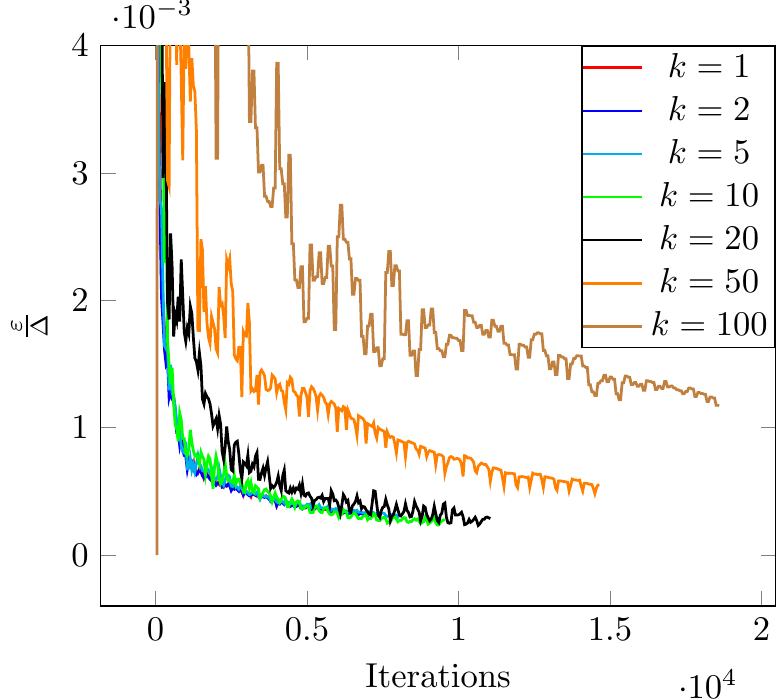}
		\caption{K3-6. Convergence in number of iterations for CFR-Jr with different reconstruction rates}
		\label{fig:reconstruction_rate1}
	\end{minipage}
	\hspace{0.5cm}
	\begin{minipage}{.3\textwidth}
		\centering
		\includegraphics[width=\textwidth]{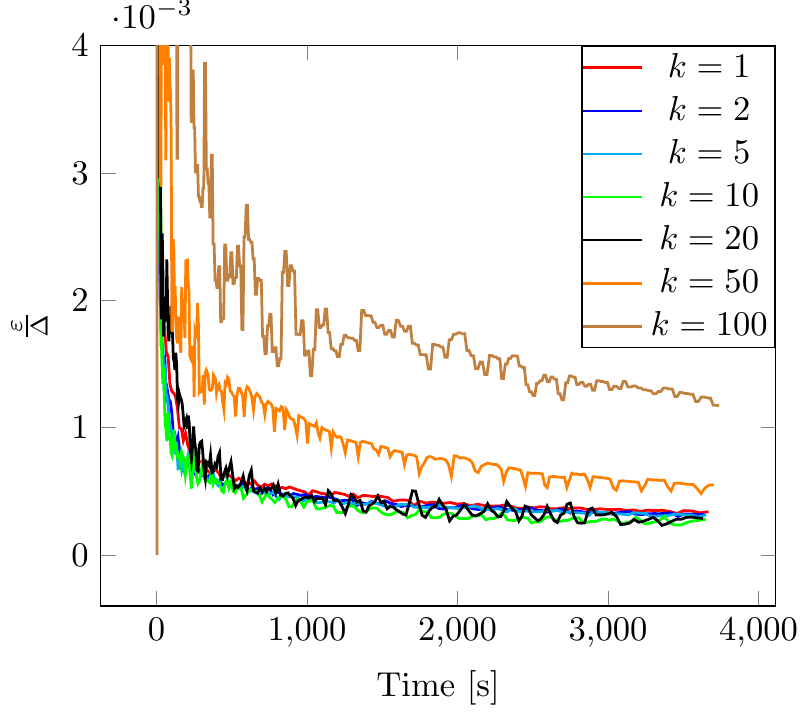}
		\caption{K3-6. Convergence in run time (seconds) for CFR-Jr with different reconstruction rates}
		\label{fig:reconstruction_rate2}
	\end{minipage}
	\hspace{0.5cm}
	\begin{minipage}{.3\textwidth}
		\vspace{2mm}
		\centering
		\includegraphics[width=\textwidth]{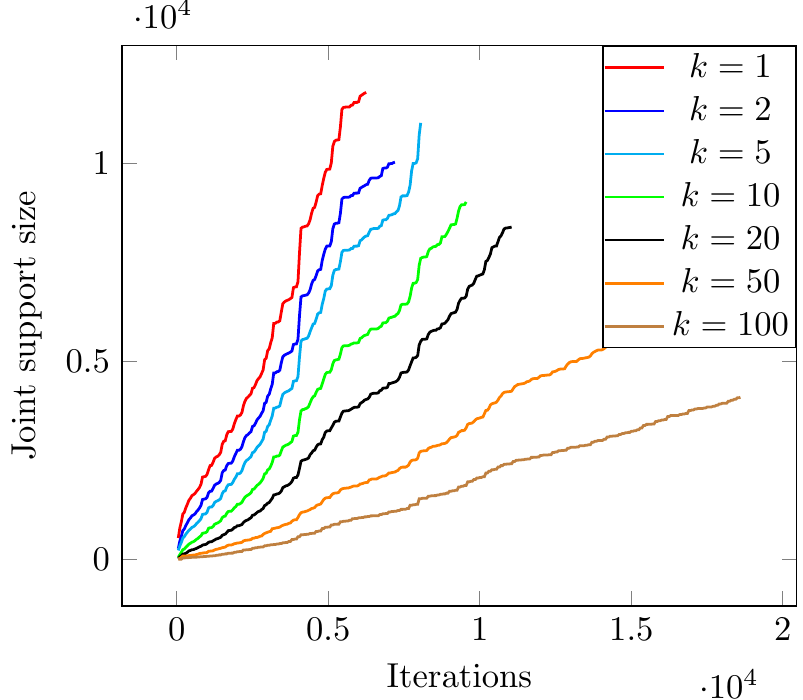}
		\caption{K3-6. Size of the support of the joint strategy obtained from CFR-Jr with different reconstruction rates}
		\label{fig:reconstruction_rate3}
	\end{minipage}	
\end{figure*}

\begin{figure*}[t]
	\begin{minipage}{.3\textwidth}
		\centering
		\includegraphics[width=\textwidth]{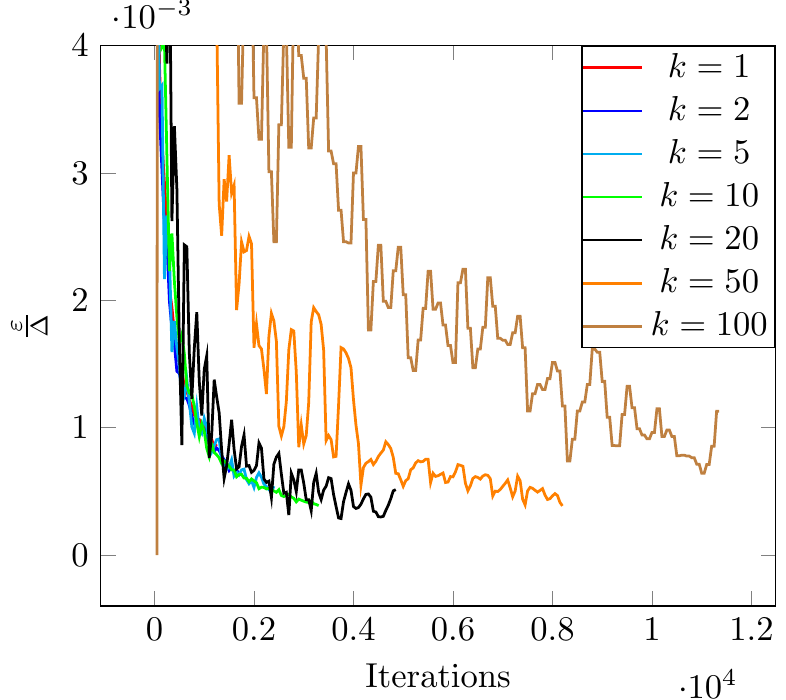}
		\caption{K3-10. Convergence in number of iterations for CFR-Jr with different reconstruction rates}
		\label{fig:reconstruction_rate1_k_3_10}
	\end{minipage}
	\hspace{0.5cm}
	\begin{minipage}{.3\textwidth}
		\centering
		\includegraphics[width=\textwidth]{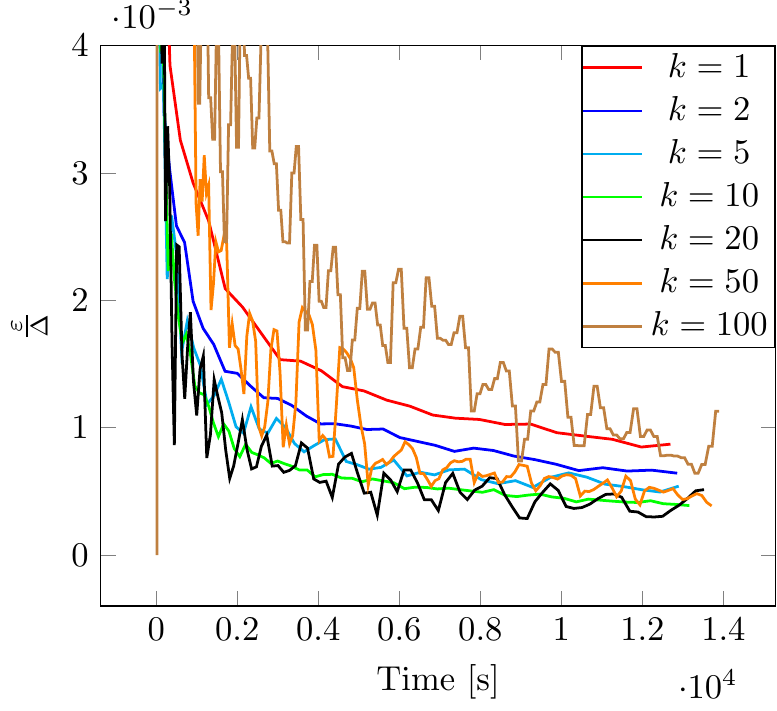}
		\caption{K3-10. Convergence in run time (seconds) for CFR-Jr with different reconstruction rates}
		\label{fig:reconstruction_rate2_k_3_10}
	\end{minipage}
	\hspace{0.5cm}
	\begin{minipage}{.3\textwidth}
		\vspace{1.8mm}
		\centering
		\includegraphics[width=\textwidth]{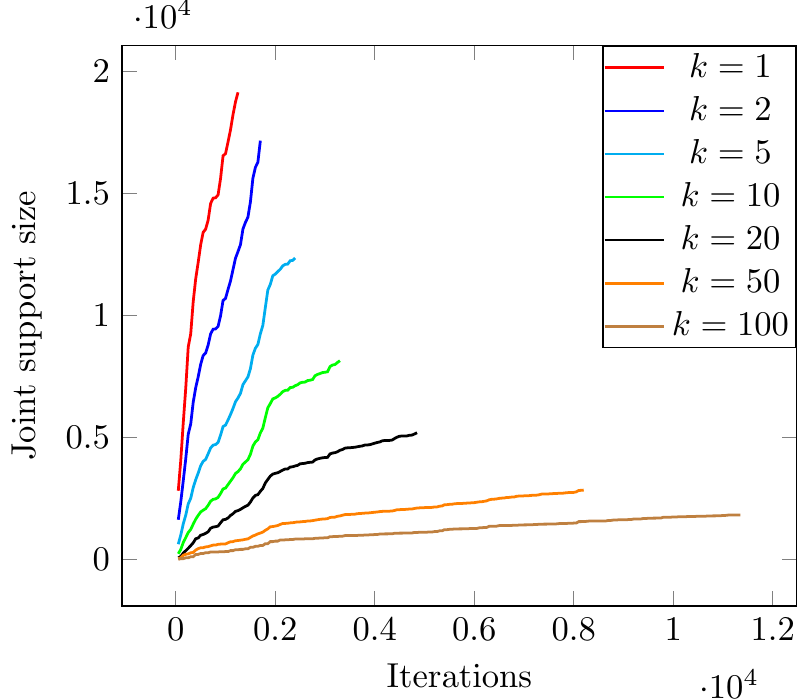}
		\caption{K3-10. Size of the support of the joint strategy obtained from CFR-Jr with different reconstruction rates}
		\label{fig:reconstruction_rate3_k_3_10}
	\end{minipage}	
\end{figure*}